\documentclass[journal,twoside,web]{ieeecolor}
\usepackage{generic}
\usepackage{cite}

\usepackage{amsmath,amssymb,amsfonts,amsthm}
\usepackage{algorithmic}
\usepackage{algorithm}
\usepackage{graphicx}
\usepackage{textcomp}
\usepackage{svg}
\usepackage{subfig}

\newtheorem{theorem}{Theorem}
\newtheorem{definition}{Definition}

\newtheorem{proposition}{Proposition}
\newtheorem{lemma}{Lemma}
\newtheorem{corollary}{Corollary}
\newtheorem{remark}{Remark}

\newtheorem{assumption}{Assumption}

\def\BibTeX{{\rm B\kern-.05em{\sc i\kern-.025em b}\kern-.08em
		T\kern-.1667em\lower.7ex\hbox{E}\kern-.125emX}}
\begin{document}
\title{Stability Analysis and Intervention Strategies on a Coupled SIS Epidemic Model with \\ Polar Opinion Dynamics}
\author{Qiulin Xu, Tatsuya Masada, and Hideaki Ishii, \IEEEmembership{Fellow, IEEE}
\thanks{This work was supported in the part by JSPS under Grant-in-Aid for Scientific Research Grant No. 22H01508, and in the part by JST SPRING, Japan Grant Number JPMJSP2106 and JPMJSP2180.}% <-this % stops a space
\thanks{$^{1}$Qiulin Xu and Tatsuya Masada are with the Department of Computer Science, Institute of Science Tokyo (formerly, Tokyo Institute of Technology), Yokohama, Japan.
	{\tt\small xu.q.76b5@m.isct.ac.jp, m.tatsuya0210@gmail.com}}%
\thanks{$^{2}$Hideaki Ishii is with the Department of Information Physics and Computing, The University of Tokyo, Tokyo, Japan.
	{\tt\small hideaki\_ishii@ipc.i.u-tokyo.ac.jp}}%
	}

\maketitle

\begin{abstract}
%This paper investigates the spread of infectious diseases within a networked community, integrating the dynamics of both epidemic transmission and public opinion. We propose a novel discrete-time networked SIS (Susceptible-Infectious-Susceptible) epidemic model coupled with opinion dynamics that accounts for the presence of stubborn agents with extreme views. The model captures the interplay between an individual's perception of epidemic severity and the actual spread of the disease. It is capable to offer a more comprehensive understanding of epidemic dynamics in a socially interconnected environment. We introduce the concept of the SIS-opinion reproduction number to assess the severity of the epidemic and analyze the conditions for natural disease eradication and the global stability of the endemic equilibrium. Furthermore, we explore the potential for controlling the epidemic through opinion intervention strategies, offering a framework for policymakers to design effective epidemic prevention measures. Numerical examples validate our theoretical findings and demonstrate the practical implications of the proposed model.
This paper investigates the spread of infectious diseases within a networked community by integrating epidemic transmission and public opinion dynamics. We propose a novel discrete-time networked SIS (Susceptible-Infectious-Susceptible) epidemic model coupled with opinion dynamics that includes stubborn agents with biased views. The model captures the interplay between perceived and actual epidemic severity, offering insights into epidemic dynamics in socially interconnected environments. We introduce the SIS-opinion reproduction number to assess epidemic severity and analyze conditions for disease eradication and the global stability of endemic equilibria. Additionally, we explore opinion-based intervention strategies, providing a framework for policymakers to design effective prevention measures. Numerical examples are provided to illustrate our theoretical findings and the model's practical implications.
\end{abstract}


\section{Introduction}
Since the discovery of infectious diseases, researchers across fields have contributed to theoretical epidemiology, developing mathematical models to accurately describe disease propagation \cite{pare2020analysis,liu2019analysis,she2022networked,wang2022resilient,wu2023switching,wu2019adaptive,wang2025TAC,Xie2023auto,yu2025auto}. Epidemic models have become crucial over recent decades for predicting epidemic evolution and guiding public health policies \cite{zino2021analysis}. The global COVID-19 outbreak, with its significant societal impact, has led researchers to explore various factors affecting epidemic spread, including environmental influences (e.g., water supply \cite{randazzo2020sars} and public transport \cite{karatas2022transportation}), public awareness, and opinion dynamics \cite{teslya2022effect}. Consequently, developing more complex, accurate models for epidemic dynamics remains a pressing issue.

In modeling disease spreading, most models adopt compartmental structures, segmenting populations by health statuses. A typical example is the SIR model, which classifies the individuals as Susceptible, Infectious, or Recovered. Individuals in the recovered state do not get infected again, leading to the eventual disappearance of the epidemic. This model is well-suited to diseases that confer lifelong immunity, such as chickenpox, where the primary goal is infection peak control \cite{wang2022resilient,di2020covid}. Another common model, the SIS, assumes that individuals can be reinfected after recovery. Due to the continuous mutation of infectious diseases like influenza and COVID-19, their outbreaks are not transient. For these diseases, the SIS model is clearly more appropriate \cite{pare2020analysis,liu2019analysis}.

One effective method for constructing dynamic models that capture the epidemic process of infectious diseases is to use networked models, where each node represents a human community, and edges between nodes represent pathways for diseases spreading between communities. In networked infectious disease models, transmission rates, recovery rates, and network structures are all critical for characterizing the epidemic process. Recently, models have been studied, considering the influence of human awareness \cite{paarporn2017networked} and the impact of opinion interactions on disease threat perception \cite{lin2021discrete}. The study in \cite{lin2021discrete} considers the coupling of a networked SIS model with the fundamental DeGroot opinion dynamics model \cite{degroot1974reaching}, in which each agent updates its opinion by taking a weighted average of others' opinions, leading to a simple consensus in strongly connected network structures. In \cite{she2022networked}, the relationships of cooperation and opposition among opinions are further considered. The works \cite{wang2022resilient} and \cite{wang2021suppressing} examine cases where infected agents exhibit abnormal behaviors in multi-agent consensus problems. However, in reality, the tendency for individual opinion changes are often state-dependent, manifesting as different degrees of stubbornness among holders of different opinions \cite{amelkin2017polar}. For example, those who take the epidemic seriously tend to maintain caution and preventive measures even as the epidemic begins to subside \cite{organisation2022first}. A recent work \cite{yu2024individuals} introduces an analytical framework that simulates the interplay between opinion dynamics and epidemic spread, providing qualitative insights into how opinion polarization influences the transmission of infectious diseases. Inspired by the above literature, this study theoretically studies a discrete-time networked SIS model considering stubborn agents with extreme opinions, based on mass panic theory in social psychology \cite{mawson2017mass}.

%One effective method for modeling infectious disease dynamics is through networked models, where each node represents a community and edges represent transmission pathways. In these models, the transmission rate, recovery rate, and network structure are critical in characterizing the epidemic process. Recent studies consider the effects of human awareness \cite{paarporn2017networked} and opinion interactions regarding disease threat \cite{lin2021discrete}. The study in \cite{lin2021discrete} couples a networked SIS model with the DeGroot opinion dynamics model \cite{degroot1974reaching}, where agents update opinions through weighted averages, leading to consensus in strongly connected networks. In \cite{she2022networked}, opinion cooperation and opposition are further examined. The works \cite{wang2022resilient} and \cite{wang2021suppressing} address abnormal behaviors in consensus problems among agents with infections. However, opinion changes in real contexts are often state-dependent, with varying stubbornness among agents \cite{amelkin2017polar}. For example, individuals who take the epidemic seriously will tend to maintain a cautious attitude and preventive measures even after the epidemic begins to subside \cite{organisation2022first}. Thus, this study proposes a discrete-time networked SIS model incorporating stubborn agents with extreme opinions, grounded in mass panic theory from social psychology \cite{mawson2017mass}.

The contribution of this work is threefold: First, we construct a new networked SIS model that considers the interaction between agent opinions and infectious diseases, taking into account the presence of stubborn agents with extreme opinions. Second, we define an SIS-opinion reproduction number to quantify epidemic severity. In addition to analyzing conditions for natural disease eradication as commonly done in traditional epidemiological studies, we highlight that this work provides sufficient conditions for the existence and global stability of an endemic equilibrium in more severe outbreaks. Third, we discuss the possibility of controlling the epidemic through opinion intervention, summarizing an algorithm as a reference for policymakers in formulating epidemic prevention strategies under different scenarios.

This paper is organized as follows. Section \ref{Sec:Preliminary} introduces key preliminaries, including graph structures, opinion dynamics, and epidemic models, and proposes a networked SIS epidemic model that includes opinion dynamics considering stubbornness. Section \ref{Sec:Results} analyzes the properties of the proposed model and the behavior of the epidemic spreading process. The suppression of epidemics by influencing the opinions and the practical epidemic countermeasures are also explored in Section \ref{Sec:Results}. Numerical examples in Section \ref{Sec:Simulation} illustrate the analytical results, and the paper concludes in Section \ref{Sec:Conclusion}. Compared to the preliminary version \cite{xu2024analysis}, this paper contains all the proofs of the theoretical results and further discussions.

\textit{Notation:} Let $[n]$ denote the set $\{1,2, \ldots, n\}$ for any positive integer $n$. Denote by $\mathbb{R}^n$ and $\mathbb{R}^{n \times n}$ the $n$-dimensional Euclidean space and the set of $n \times n$ real matrices, respectively. The superscript ``${\top}$'' stands for transposition of a matrix. Denote by $\rho(\cdot)$, $\|\cdot\|$, and $\|\cdot\|_{\infty}$ the spectral radius, Euclidean norm, and infinity norm of a matrix, respectively. Denote by $\boldsymbol{0}_n$ and $\boldsymbol{1}_n$ the all-zero and all-one vectors in $\mathbb{R}^n$, respectively, and $I_{n}$ denotes $n \times n$ identity matrix. For any matrix $M \in \mathbb{R}^{n \times n}$, we denote its $(i,j)$-th entry by $M_{ij}$. For any two vectors $x, y \in \mathbb{R}^n$, we simply write $x > y$ if $x_i > y_i, \forall i \in [n]$.

\section{Preliminaries and Problem Formulation} \label{Sec:Preliminary}
In this section, we first introduce some standard definitions and key concepts from graph theory used in this paper. Then, we consider a group of communities where the infectious disease spreads through a physical network among communities, and the community opinions about the disease evolve on a social network. To this end, we construct a coupled network model to describe the co-evolution of these two processes.

\subsection{Graph Theory}
Consider a system consisting of $n$ agents represented by a directed graph $\mathcal{G}=(\mathcal{V}, \mathcal{E})$. Here, $\mathcal{V}=[n]$ is the set of nodes, and $\mathcal{E} \subset \mathcal{V} \times \mathcal{V}$ is the set of edges. In this study, each node is referred to as an agent, representing a human community. An edge $(j, i) \in \mathcal{E}$ indicates that information can be transmitted from agent $j$ to agent $i$. Particularly, $(i, i) \in \mathcal{E}$ is called a self-loop. In a directed graph, $(j, i) \in \mathcal{E}$ does not necessarily imply $(i, j) \in \mathcal{E}$. When $(j, i) \in \mathcal{E}$, agent $j$ is said to be a neighbor of agent $i$. The set of neighbors is denoted as $\mathcal{N}_i=\{j:(j, i) \in \mathcal{E}\}$ with its cardinality denoted as $d_i=\left|\mathcal{N}_i\right|$.

Let $A=[a_{ij}]_{n \times n}$ be the adjacency matrix of $\mathcal{G}$, where $a_{ij} > 0$ if $(j, i) \in \mathcal{E}$ and $a_{ij} = 0$ otherwise. The graph Laplacian is a matrix representing the properties of the graph structure and is defined as $L=[l_{ij}] \in \mathbb{R}^{n \times n}$, where $l_{i i}=$ $\sum_{j=1, j \neq i}^n a_{i j}$ and $l_{i j}=-a_{i j}$ for $i \neq j$; notice that each row of $L$ sums to zero. 
%Therefore, for $x \in \mathbb{R}^n$, the following relationship holds:
%$$
%\begin{aligned}
%	{[L x]_i } & =\sum_{j=1}^n l_{i j} x_j=l_{i i} x_i+\sum_{j=1, j \neq i}^n l_{i j} x_j=\left(\sum_{j=1, j \neq i}^n a_{i j}\right) x_i+\sum_{j=1, j \neq i}^n\left(-a_{i j}\right) x_j \\
%	& =\sum_{j=1, j \neq i}^n a_{i j}\left(x_i-x_j\right)=\sum_{j \in \mathcal{N}_i} a_{i j}\left(x_i-x_j\right).
%\end{aligned}
%$$

\subsection{Opinion Dynamics Model}
Opinion dynamics is widely studied as a theory to analyze the process of opinion evolution and formation when agents interact and exchange opinions on a social network \cite{proskurnikov2017tutorial}. Generally, a community does not always maintain its initial opinion. Interaction is constant in that each community evaluates its own opinion and the opinions of its neighbors, continually updating its own view while incorporating its own preferences and characteristics. We consider opinion evolution in a social network of $n$ communities, represented as a directed graph $\mathcal{G}_S=\left(\mathcal{V}, \mathcal{E}_S\right)$, with $\mathcal{N}_i^S$ denoting the set of neighbors of community $i$ in graph $\mathcal{G}_S$. The relative influence of community $j$'s opinion upon community $i$ is measured by the weight $w_{ij} \in [0,1)$. Assume that $\sum_j w_{i j}=1$ for all $i \in[n]$. Then the matrix $W = [w_{ij}]\in \mathbb{R}^{n \times n}$ and the Laplacian matrix $\bar{L} =[\bar{l}_{ij}]= I_n - W$ represent the row-stochastic adjacency matrix and the Laplacian matrix of this social network, respectively.

We use the parameter $z_i(k)$ to represent the opinion of community $i$ regarding the severity of the epidemic at time $k$, and this parameter takes a value in the range $[0, 1]$. When $z_i(k)=1$, it means that community $i$ takes the epidemic highly seriously, whereas $z_i(k)=0$ indicates that community $i$ is extremely complacent about the epidemic. Naturally, an individual's or community's receptiveness to external information and their inclination to change their current stance are closely related to their current opinion. Some opinions are held more stubbornly, while others are more flexible and easily influenced. This phenomenon is known as attitude polarization in social psychology \cite{miller1993attitude}. In this study, we follow the theory of mass panic \cite{mawson2017mass}. According to this theory, opinions that view the epidemic as severe tend to reinforce themselves, while opinions that underestimate the epidemic are more susceptible to panic-induced influence. 

In view of these points, we consider the following polar opinion dynamical model with \textit{stubborn positives} from \cite{amelkin2017polar}:
\begin{equation} \label{original-opinion}
	z_i(k+1) = z_i(k)+\Big(1-z_i(k)\Big) \sum_{j \in \mathcal{N}_i^S} w_{i j}\Big(z_j(k)-z_i(k)\Big).
\end{equation}
Under this model, it is easy to see that when $z_i(k)=1$, agent $i$ will be completely stubborn in its opinion. Conversely, when $z_i(k)=0$, this model degenerates into the classical DeGroot model \cite{degroot1974reaching}, fully accepting the influence of external opinions.

\subsection{Networked SIS Epidemic Model}
The SIS epidemic model is widely known as a basic model for infectious diseases. In the SIS model, the entire population is divided into two states: Susceptible and Infectious. The transition from Susceptible to Infectious is determined by the infection rate $\beta \in [0,1]$, and the transition from Infectious to Susceptible is determined by the recovery rate $\delta \in [0,1]$.

%\begin{figure}[h]
%	\begin{center}
%		\includegraphics[width=2.0in]{SIS.png}    % The printed column  
%		\caption{Concept of the SIS epidemic model.}  % width is 8.4 cm.
%		\label{fig:SIS}                                 % Size the figures 
%	\end{center}                                 % accordingly.
%\end{figure}

In this study, we elaborate on the networked version of the SIS epidemic model proposed by \cite{pare2020analysis}. Consider the spread of the epidemic within the same $n$ communities involved in the opinion dynamics. However, the virus transmission clearly does not occur via the social network $\mathcal{G}_S$. Instead, we capture the disease spreading on a physical connectivity network $\mathcal{G}_P=\left(\mathcal{V}, \mathcal{E}_P\right)$, also represented by a directed graph, with $\mathcal{N}_i^P$ denoting the set of neighbor nodes of community $i$ in graph $\mathcal{G}_P$. The model is also in the discrete-time domain and is given as follows:
\begin{equation} \label{original-epidemic}
	x_i(k+1)=\left(1-\delta_i\right) x_i(k)+\left(1-x_i(k)\right) \sum_{j \in \mathcal{N}_i^P} \beta_{i j} x_j(k).
\end{equation}
Here, $x_i(k) \in[0,1]$ represents the proportion of the population that is infectious in community $i$ at time $k$. The parameter $\delta_i \in[0,1]$ is the recovery rate, representing the proportion of the infectious individuals in community $i$ who recover at the next time step. Moreover, $\beta_{i j} \in[0,1]$ is called the infection rate, indicating the rate at which the infection spreads from community $j$ to community $i$. Thus, the term $-\delta_i x_i(k)$ represents the recovery rate of agent $i$, and the term $(1-$ $\left.x_i(k)\right) \sum_{j \in N_i} \beta_{i j} x_j(k)$ represents the infection rate influenced by neighboring communities.

\subsection{Coupled SIS-Opinion Dynamical Model}
After introducing the opinion dynamics and SIS epidemic spread model separately, it is natural to consider that, in reality, besides their individual evolutions as described in (\ref{original-opinion}) and (\ref{original-epidemic}), their mutual influence must also be considered. To this end, in this section, we establish a coupled epidemic-opinion model.

First, based on risk perception theory \cite{sjoberg2020explaining}, we note that, beyond the polar opinion formation described in (\ref{original-opinion}), communities' opinions on the epidemic are also influenced by their actual epidemic data. For instance, during the extreme COVID-19 spread in Japan (with Tokyo's daily infections once exceeding 40,000), both government emergency declarations and public awareness of preventive measures peaked (\textit{stubborn positives}) \cite{lu2023daily}. Even as other countries lifted lockdowns and relaxed preventive measures, Japan remained vigilant. However, as the epidemic subsided domestically, the government and public gradually lowered their risk perception and eased protective measures in response to the observed situation \cite{organisation2022first}. 

Taking account of these aspects, we consider the following improved opinion dynamics model:
	\begin{align}
		z_i(k+1) &= \theta_i x_i(k)+\left(1-\theta_i\right)\Big(z_i(k) \nonumber\\
		&\ +\big(1-z_i(k)\big) \sum_{j \in \mathcal{N}_i} w_{i j}\big(z_j(k)-z_i(k)\big)\Big). \label{opinion}
	\end{align}
Here, $\theta_i \in(0,1)$ is a constant, and $w_{i i}$ and $w_{i j}$ are positive weighting coefficients such that $w_{i i}+\sum_{j \in \mathcal{N}_i} w_{i j}=1$. Then, we can transform (\ref{opinion}) into
	\begin{align}
		z_i(k+1)&=z_i(k)+\theta_i\left(x_i(k)-z_i(k)\right) +\left(1-\theta_i\right)\left(1-z_i(k)\right) \nonumber\\
		&\ \times \sum_{j \in \mathcal{N}_i} w_{i j}\left(z_j(k)-z_i(k)\right). \label{transformed-opinion}
	\end{align}
From this equation, we can see that the dynamics of $z_i(k)$ has two properties. First, when the opinion of community $i$ aligns with the current infection state $x_i(k)$, the second term in (\ref{transformed-opinion}) is zero. In this case, the opinion dynamics is only influenced by the third term, which represents polar dynamics and tends to form stubborn positives. Second, when there is a discrepancy between the opinion $z_i(k)$ and the infection state $x_i(k)$, opinion evolution is influenced by both the polar dynamics, which exhibits stubbornness for higher values of $z_i$ (the third term in (\ref{transformed-opinion})) and a tendency to align with the actual infection situation in the community (the second term in (\ref{transformed-opinion})).

Next, the health belief model \cite{green2020health} and the theory of planned behavior \cite{ajzen2020theory} suggest that a community's perception of the threat posed by an infectious disease and their intention to mitigate its consequences directly determine their adoption of protective measures. For instance, when a community perceives the infectious disease as more dangerous, measures such as masking, social distancing, and vaccination are more likely to be actively implemented, enhancing the recovery rate and reducing infection through interactions with other communities. Thus, the actual recovery and infection rates should be functions of the opinion $z_i(k)$. 

To reflect these aspects, we propose a novel SIS epidemic model incorporating the agents' opinions. Specifically, the dynamics of $x_i(k)$ for agent $i$ is rewritten as
\begin{align}
	&x_i(k+1) \nonumber\\
	&\ = x_i(k)-\Big(\delta_{\min }+\left(\delta_i-\delta_{\min }\right) z_i(k)\Big) x_i(k) \nonumber \\
	&\ \ +\Big(1-x_i(k)\Big)\! \sum_{j \in \mathcal{N}_i}\!\Big(\beta_{i j}-\left(\beta_{i j}-\beta_{\min }\right) z_i(k)\Big) x_j(k), \label{epidemic}
\end{align}
where $\delta_{\min }$ and $\beta_{\min }$ represent the minimum recovery rate and the minimum transmission rate, respectively. If $z_i(k)=0$, agent $i$ perceives no threat from the virus and takes no actions to protect themselves, thus being maximally exposed to the infection. Conversely, if $z_i(k)=1$, agent $i$ perceives the virus as extremely severe and minimizes contact with other agents while seeking medical treatments to the maximum extent.

Additionally, we pose the following natural restrictions related to the parameters and graphs throughout this paper:
\begin{assumption} \label{ass:1}
	For each $i \in \mathcal{V}$, it holds $0 < \delta_{\min } \leq \delta_i \leq 1$, and for all $j \in \mathcal{N}_i$, $0 < \beta_{\min } \leq \beta_{i j}$ and $\sum_{j=1}^n \beta_{i j} \leq 1$. Furthermore, the communities' networks $\mathcal{G}_{P}$ and $\mathcal{G}_{S}$ are both strongly connected. 
\end{assumption}

Let $x(k)$ and $z(k)$ be $n$-dimensional vectors whose $i$th components are $x_i(k)$ and $z_i(k)$, respectively. Then, the update equations (\ref{opinion}) and (\ref{epidemic}) can be rewritten in a combined form:
\begin{equation}\label{coupled}
	\begin{bmatrix}
		x(k+1) \\
		z(k+1)
	\end{bmatrix}=\begin{bmatrix}
		\bar{K}(k) & 0_{n \times n} \\
		\Theta & \left(I_n-\Theta\right)(W+Z(k) \bar{L})
	\end{bmatrix} \begin{bmatrix}
		x(k) \\
		z(k)
	\end{bmatrix}, 
\end{equation}
where $\bar{K}(k)=I_n-\left[\Delta_{\min }+\left(\Delta-\Delta_{\min } \right) Z(k)\right]+\left(I_n-X(k)\right)\left[B-Z(k)\left(B-B_{\min}\right)\right]$, $\Delta=\text{diag}\left\{\delta_i\right\}$, $\Delta_{\min}=\delta_{\min} I_n$, $Z(k)=\text{diag}\{z(k)\}$, $X(k)=\operatorname{diag}\{x(k)\}$, $ \Theta=\text{diag}\{\theta_i\}$, $B=\left[\beta_{i j}\right]_{n \times n}$, and $B_{\min}=\beta_{\min} A$.

\subsection{Problem of Interests}
When discussing mathematical models of infectious diseases, the concept of the basic reproduction number $R_0$ is crucial. In epidemiology, the basic reproduction number refers to the average number of secondary infections produced by a single infected individual in a completely susceptible population, in the absence of any interventions. In other words, if $R_0 > 1$, the growth rate of the infected population is positive, and the epidemic will spread. If $R_0< 1$, the growth rate of the infected population is negative, and the infection will eventually disappear without any intervention. 

For the coupled SIS-opinion model (\ref{coupled}), exploring an effective reproduction number is of importance. It must account for the influence of opinions. This motivates us to propose the concept of the \textit{SIS-opinion reproduction number} in this paper and to analyze the stability of system equilibria under different conditions, including both the disease-free healthy equilibrium and the endemic equilibrium where the disease persists. We will also discuss the potential of using social networks for opinion interventions to control the epidemic.

\section{Main Results} \label{Sec:Results}
In this section, we present analysis results on the properties of the proposed SIS-opinion model (\ref{coupled}). Initially, we introduce an important lemma for solving the problem. Further, we analyze various properties of the equilibria of the SIS-opinion model (\ref{coupled}). Based on this, we use the opinion-dependent reproduction number to characterize the behavior of our model.

\subsection{Well-Posedness}
To ensure that our SIS-opinion model (\ref{coupled}) is well-posed, we must verify that its solutions always remain within the state space $[0,1]^n$. The following lemma demonstrates this property.
\begin{lemma} \label{well-posedness}
	For any initial states $x_i(0), z_i(0) \in[0,1]$, $\forall i \in [n]$, (i) it holds $x_i(k), z_i(k) \in[0,1]$ for $k \geq 0$. (ii) If $\exists j \in [n]$ such that $x_j(0) \in (0, 1]$, then it holds that for some $k^{\prime}>0$, $\forall i \in [n]$, $x_i(k)\in(0,1]$ for $k \geq k^{\prime}$.
\end{lemma}

\begin{proof}
	We prove by induction. Suppose that for $k$, it holds $x_i(k), z_i(k) \in[0,1]$. Then from (\ref{epidemic}), one obtains
	$$
	\begin{aligned}
		x_i(k&+1) =\ x_i(k)\left[1-\left(\delta_{\min }+\left(\delta_i-\delta_{\min }\right) z_i(k)\right)\right] \\
		&+\left(1-x_i(k)\right) \sum_{j \in \mathcal{N}_i}\left(\beta_{i j}-\left(\beta_{i j}-\beta_{\min }\right) z_i(k)\right) x_j(k).
	\end{aligned}
	$$
	It is evident that $x_i(k+1)$ is a convex combination of $\left[1-\left(\delta_{\min }+\left(\delta_i-\delta_{\min }\right) z_i(k)\right)\right]$ and $\sum_{j \in \mathcal{N}_i}\left(\beta_{i j}-\left(\beta_{i j}-\beta_{\min }\right) z_i(k)\right) x_j(k)$, both within $[0,1]$ for all $i$ by Assumption \ref{ass:1}. Hence, it follows that $x_i(k+1) \in[0,1]$. Furthermore, if there exists $j \in [n]$ such that $x_j(0) > 0$, then by the strong connectivity in Assumption \ref{ass:1}, both terms will be in $(0,1]$ after at most $n-1$ steps. Therefore, for some $k^{\prime} > 0$, we have $x_i(k)\in(0,1]$ for $k \geq k^{\prime}$. 
	
	On the other hand, in (\ref{opinion}), we notice that $z_i(k)+\left(1-z_i(k)\right) \sum_{j \in \mathcal{N}_i} w_{i j}\left(z_j(k)-z_i(k)\right)$ is a convex combination of $1$ and $\sum_{j \in \mathcal{N}_i} w_{i j}\left(z_j(k)-z_i(k)\right)$, and the latter term is in $[0,1]$ for all $i$ by Assumption \ref{ass:1}. Furthermore, $z_i(k+1)$ is also a convex combination of $z_i(k)+\left(1-z_i(k)\right) \sum_{j \in \mathcal{N}_i} w_{i j}\left(z_j(k)-z_i(k)\right)$ and $x_i(k)$. Therefore, we conclude that $z_i(k+1) \in[0,1]$.
\end{proof}	

\subsection{Forms of the Equilibria}
We now introduce some properties of the equilibra of our system. Let $s^*=\left(x^*, z^*\right)$ denote an equilibrium of (\ref{coupled}). We refer to $s^*$ as \textit{healthy} if the epidemic has completely disappeared, i.e., $x^*=\textbf{0}_n$. Otherwise, we call it \textit{endemic}. Additionally, if all communities eventually reach an agreement on their opinions about the epidemic, i.e., $z_i^*=z_j^*$, $\forall i, j \in [n]$, we refer to the equilibrium as \textit{consensus}. 

Consider an equilibrium $s^*=\left(x^*, z^*\right)$ of system (\ref{coupled}). It is easy to see that $\delta_{\min } I_n+\left(\Delta-\delta_{\min } I_n\right) Z^*$ is a positive diagonal matrix. Hence by (\ref{epidemic})
\begin{align}
	x^*=& \ \left[\delta_{\min } I_n+\left(\Delta-\delta_{\min } I_n\right) Z^*\right]^{-1}\left(I_n-X^*\right) \nonumber\\
	&\times \left[B-Z^*\left(B-\beta_{\min } A\right)\right] x^* . \label{x*}
\end{align}
Moreover, we can verify that $I_n-\left(I_n-\Theta\right)\left(W+Z^* \bar{L}\right)$ is invertible. It thus follows from (\ref{opinion}) that
\begin{equation}
	z^*=\left[I_n-\left(I_n-\Theta\right)\left(W+Z^* \bar{L}\right)\right]^{-1} \Theta x^*. \label{z*}
\end{equation}

%\begin{proposition} \label{form of equilibrium}
%	If $s^*=\left(x^*, z^*\right)$ is an equilibrium point of system (\ref{coupled}), then it satisfies
%	$$
%	\begin{aligned}
	%		& x^*=\left[\delta_{\min } I_n+\left(D-\delta_{\min } I_n\right) Z^*\right]^{-1}\left(I_n-X^*\right)\left[B-Z^*\left(B-\beta_{\min } A\right)\right] x^* \\
	%		& z^*=\left[I_n-\left(I_n-\Theta\right)\left(W+Z^* \bar{L}\right)\right]^{-1} \Theta x^*
	%	\end{aligned}
%	$$
%\end{proposition}
%\begin{proof}
%	Since $s^*$ is an equilibrium point, one has $x(k+1)=x(k)=x^*, z(k+1)=z(k)=z^*$. It then follows from (\ref{epidemic}) that 
%	$$
%	x^*=x^*-\left[\delta_{\min } I_n+\left(D-\delta_{\min } I_n\right) Z^*\right] x^*+\left(I_n-X^*\right)\left[B-Z^*\left(B-\beta_{\min } A\right)\right] x^*.
%	$$
%	Note that $\left[\delta_{\min } I_n+\left(D-\delta_{\min } I_n\right) Z^*\right]$ is a positive diagonal matrix, hence
%	$$
%	x^*=\left[\delta_{\min } I_n+\left(D-\delta_{\min } I_n\right) Z^*\right]^{-1}\left(I_n-X^*\right)\left[B-Z^*\left(B-\beta_{\min } A\right)\right] x^*.
%	$$
%	Moreover, by (\ref{opinion}), it follows that
%	$$
%	z^*=\Theta x^*+\left(I_n-\Theta\right)\left(W+Z^* \bar{L}\right) z^*.
%	$$
%	It is not difficult to verify that $I_n-\left(I_n-\Theta\right)\left(W+Z^* \bar{L}\right)$ is invertible. One thus obtains
%	$$
%	z^*=\left[I_n-\left(I_n-\Theta\right)\left(W+Z^* \bar{L}\right)\right]^{-1} \Theta x^*,
%	$$
%	which completes the proof.
%\end{proof}

We observe from (\ref{x*}) and (\ref{z*}) that $s^*=(\mathbf{0}_n, \mathbf{0}_n)$ is a special healthy-consensus equilibrium, representing the ideal state where the epidemic is eradicated on both physical and social levels. This is the key focus of our next section. 

The following proposition shows that, apart from the healthy equilibrium, specific conditions of system parameters must be met for communities to reach a consensus on opinions.

\begin{proposition} \label{prop:persistence}
	For a nonzero equilibrium $s^*=\left(x^*, z^*\right)$ of system (\ref{coupled}), it holds $\mathbf{0}_n<s^*<\mathbf{1}_n$, and the opinions reach a consensus (i.e., $z^*=a \mathbf{1}_n, a \in(0,1]$) only if 
	\begin{align} \label{nonzero consensus}
		&\delta_{\min}+\left(\delta_i-\delta_{\min }\right) a \nonumber\\
		&\quad =(1-a) \sum_{j \in \mathcal{N}_i}\left(\beta_{i j}-\left(\beta_{i j}-\beta_{\min }\right) a\right), \ \forall i \in[n].
	\end{align}
\end{proposition}
\begin{proof}
	We first show that $\mathbf{0}_n<s^*<\mathbf{1}_n$. Suppose that $s^*$ is nonzero but there exists a community $i \in [n]$ with $x_i^*=0$. Substituting $x_i^*=0$ into the epidemic model (\ref{epidemic}), since $\beta_{i j}-\left(\beta_{i j}-\beta_{\min }\right) z_i(k)$ is strictly positive, we have $x_j^*=0$ for all $j \in \mathcal{N}_i$. Since the network is strongly connected, $x^*=\mathbf{0}_n$. Then from (\ref{z*}), one can obtain that $z^* = \mathbf{0}_n$. It thus follows that $s^*=(\mathbf{0}_n, \mathbf{0}_n)$, which is a contradiction. Similar arguments lead us to $x^*<\mathbf{1}_n$ and $\mathbf{0}_n<z^*<\mathbf{1}_n$.
	
	Then to prove the condition (\ref{nonzero consensus}) of nonzero consensus, assume that $z^*=a \mathbf{1}_n$ with $a \in(0,1]$. From (\ref{opinion}), $a=\theta_i x_i^*+\left(1-\theta_i\right) a$, and it thus follows that $x_i^*=a$. From (\ref{epidemic}),
	$$
	\begin{aligned}
		a=& \ a-\left(\delta_{\min }+\left(\delta_i-\delta_{\text {min }}\right) a\right) a \\
		&+(1-a) \sum_{j \in \mathcal{N}_i}\left(\beta_{i j}-\left(\beta_{i j}-\beta_{\min }\right) a\right) a.
	\end{aligned}
	$$
	Since $a>0$, one obtains (\ref{nonzero consensus}).
\end{proof}

Proposition \ref{prop:persistence} implies that when the epidemic cannot be completely eradicated, no community can be free from the disease. Additionally, due to the heterogeneity in infection and recovery rates among communities, it may be difficult for (\ref{nonzero consensus}) to hold. This will imply that the entire society may not reach a consensus on views about the epidemic, especially in more complex social networks.

\subsection{Stability Analysis of Healthy Equilibria}
Next, to analyze the asymptotic behavior of the coupled model (\ref{coupled}) and in particular the stability of its equilibria, we introduce the opinion-dependent reproduction number. In this subsection, we first focus on the most ideal state for the society, namely the healthy equilibrium point $x^*=\mathbf{0}_n$.

For the coupled model (\ref{coupled}), we define the following opinion-dependent effective reproduction number.

\begin{definition} \label{Reproduction Number}
	(SIS-Opinion Reproduction Number) For the coupled epidemic-opinion model in (\ref{coupled}), denote by 
	\begin{equation} \label{Rz}
		R^z(k) = \rho\left(I_n-\Delta(z(k))+B(z(k))\right)
	\end{equation}
	the effective reproduction number at time $k$, where $\Delta(z(k))=\operatorname{diag}\left\{\delta_{\min}+\left(\delta_i-\delta_{\min}\right) z_i(k)\right\}$ and $B(z(k))=\left[\beta_{i j}-\left(\beta_{i j}-\beta_{\min}\right) z_i(k)\right] \in \mathbb{R}^{n \times n}$.
\end{definition}

At any given moment, the effective reproduction number (\ref{Rz}) of the epidemic depends on the opinion state $z(k)$ at that time. It is not difficult to notice that when all communities agree on the severity of the epidemic, i.e., $z(k) = z_{\max} = \textbf{1}_n$, the reproduction number will reach its minimum:
\begin{equation} \label{Rzmin}
	R^z_{\min }=\rho\left(I_n-\Delta\left(z_{\max }\right)+B\left(z_{\max }\right)\right)=\rho\left(I_n-\Delta+B_{\min }\right).
\end{equation}
Conversely, when all communities believe the epidemic is trivial and negligible, i.e., $z(k) = z_{\min} = \textbf{0}_n$, then the reproduction number will be at its maximum:
\begin{equation}\label{Rzmax}
	R^z_{\max }=\rho\left(I_n-\Delta\left(z_{\min }\right)+B\left(z_{\min }\right)\right)=\rho\left(I_n-\Delta_{\min }+B\right).
\end{equation}

Now, using $R^z(k)$ as a measure, we can analyze the spread behavior and equilibrium conditions of epidemics of varying intensities under polar opinion dynamics. We first consider the relatively low-infectivity scenario, where $R^z(k) \leq R^z_{\max } \leq 1$.

\begin{theorem} \label{healthy-consensus equilibrium}
	If $R^z_{\max} \leq 1$, the healthy-consensus equilibrium $s^*\!=\!(\mathbf{0}_n, \mathbf{0}_n)$ is asymptotically stable for all initial conditions.
\end{theorem}
\begin{proof}
	It follows from (\ref{epidemic}) that
	$$
	\begin{aligned}
		x&(k+1)-x(k) \\
		& =-\left[\Delta_{\min }+\left(\Delta-\Delta_{\min }\right) Z(k)\right] x(k) \\
		& \quad \ +\left(I-X(k)\right)\left[B-Z(k)\left(B-\beta_{\min } A\right)\right] x(k) \\
		& =-\Delta_{\min } x(k)+(I-X(k)) B x(k)-\left(\Delta-\Delta_{\min }\right) Z(k) x(k) \\
		& \quad \ -(I-X(k)) Z(k)\left(B-\beta_{\min } A\right) x(k) \\
		& \leq-\Delta_{\min } x(k)+(I-X(k)) B x(k),
	\end{aligned}
	$$
	where the inequality holds since $\Delta \geq \Delta_{\min}$, $0 \leq X(k) \leq I_n$, and $B \geq B_{\min} = \beta_{\min } A$.
	Note that $x(k+1) \geq 0$ by Lemma \ref{well-posedness}. Therefore, $x(k)$ of the original epidemic dynamics (\ref{epidemic}) is upper bounded as $x(k) \leq y(k)$, where $y(k)$ is generated by
	$$
	y(k+1)=y(k)-\Delta_{\min } y(k)+(I-\text{diag}\{y(k)\}) B y(k)
	$$
	with $y(0)=x(0)$. Then, it follows from Theorem 1 in \cite{pare2020analysis} that $y(k)$ converges to $\textbf{0}_n$ for any initial state $y(0) \in [0, 1]^n$ when $R^z_{\max} \leq 1$, and thus $x(k)$ does as well.
	
	Now we consider the dynamics of $z(k)$. By (\ref{opinion}), it holds 
	\begin{equation} \label{zk}
		z(k+1)=\left(I_n-\Theta\right)(W+Z(k) \bar{L}) z(k)+\Theta x(k).
	\end{equation}
	Note that the row sums of matrix $Z(k) \bar{L}$ are all $0$ and $W$ is a row-stochastic matrix. Since $z_i(k) \in [0,1]$, $\forall i \in [n]$, we know $W+Z(k) \bar{L} \geq 0$. Hence, $W+Z(k) \bar{L}$ is also a row-stochastic matrix. By defining $\bar{z}(k+1)=\left(I_n-\Theta\right)(W+\bar{Z}(k) \bar{L}) \bar{z}(k)$ with $\bar{Z}(k)=\text{diag}\{\bar{z}(k)\}$, it follows that 
	$
	\bar{z}_{\max }(k+1) \leq\left(1-\theta_{\min }\right) \bar{z}_{\max }(k),
	$
	where $\bar{z}_{\max }(k)=\max _{i \in[n]} \bar{z}_i(k)$ and $\theta_{\min }=\min _{i \in[n]} \theta_i$. Since $1-\theta_{\min } \in(0,1)$, $\bar{z}_{\max }(k)$ will exponentially converge to $0$ for all initial conditions. Further, it is immediate from the stability of $\bar{z}(k)$ that the original $z(k)$ in (\ref{zk}) is input-to-state stable. Hence,we have that $z(k)$ globally converges to $\textbf{0}_n$ as $x(k)$ asymptotically goes to $\textbf{0}_n$.
\end{proof}	

Theorem \ref{healthy-consensus equilibrium} reveals that a small reproduction number, namely $R^z_{\max} \leq 1$, is a sufficient condition for the global asymptotic stability of the consensus-healthy equilibrium. In other words, when the intensity of a particular epidemic is below a certain threshold, the epidemic will naturally die out without any external interventions. Simultaneously, the opinions of various communities will gradually converge to $0$ as the epidemic vanishes, just as in reality when the epidemic is over, people commonly have no reason to take it seriously anymore \cite{szczuka2021trajectory}. However, such a seemingly natural sociological trait actually implies the risk of a resurgence of the epidemic. The following proposition will theoretically demonstrate this point.

\begin{proposition} \label{healthy-unstable}
	If $R^z_{\max}>1$, the healthy state $x^*=\mathbf{0}_n$ is unstable.
\end{proposition}
\begin{proof}
	Consider the Jacobian matrix $J$ of the coupled system (\ref{coupled}) at $(x, z)$ given by
	$$
	J(x,z)=\begin{bmatrix}
		J_{11}(x,z) & J_{12}(x,z) \\
		J_{21}(x,z) & J_{22}(x,z)
	\end{bmatrix}.
	$$
	From (\ref{opinion}) and (\ref{epidemic}), each entry can be obtained as
	\begin{equation}
		\begin{aligned}
			J_{11}(x,z)= \ & (I_n - X)\left[B - Z(B-B_{\min})\right] \\
			&+ I_n - \left[\Delta_{\min} + (\Delta-\Delta_{\min})Z\right] \\
			&- \text{diag}\left\{(B-Z(B-B_{\min}))x\right\},\\
			J_{12}(x,z) = \ & - \text{diag}\{\left(I_n-X\right)(B-B_{\min})x \\
			&+ (\Delta-\Delta_{\min})x\},\\
			J_{21}(x,z) = \ & \Theta,\\
			J_{22}(x,z) = \ & (I_n-\Theta)\big(\left(I_n-Z\right)\left(WZ-ZW\right) \\
			+ [I_n &- \text{diag}\{\bar{L}z\} - \left(I_n-Z\right)(I_n-\widetilde{W}) ] \big), 
		\end{aligned} \label{Jacobian}
	\end{equation}
	where $\widetilde{W}$ is a diagonal matrix whose entries are the diagonal entries of $W$.
	
	Recall that the healthy-consensus equilibrium $s^* = (x^*, z^*) = (\mathbf{0}_n, \mathbf{0}_n)$ is the unique healthy equilibrium of the coupled system (\ref{opinion}) and (\ref{epidemic}). That is, by (\ref{z*}), when $x^*=\mathbf{0}_n$, $z^*=\mathbf{0}_n$ is the unique solution. Substituting $s^*=(\mathbf{0}_n, \mathbf{0}_n)$ into (\ref{Jacobian}), we can calculate the Jacobian matrix as 
	\begin{equation}
		J(\mathbf{0}_n, \mathbf{0}_n) = \begin{bmatrix}
			I_n-\Delta_{\min}+B & 0 \\
			\Theta & (I_n-\Theta)\widetilde{W}
		\end{bmatrix}.
	\end{equation}
	This matrix is unstable by (\ref{Rzmax}) and $R^z_{\max}>1$. Hence, by Lyapunov's indirect method, our proof is complete.
\end{proof}

The conclusions of Theorem \ref{healthy-consensus equilibrium} and Proposition \ref{healthy-unstable} confirm that $R_{\max}^z \leq 1$ is both necessary and sufficient for the global asymptotic stability of the healthy state. This implies that for more severe epidemics that cannot spontaneously disappear, the only way to end the epidemic is to reduce $R^z_{\max}$ to $1$ or lower through certain means and efforts. According to Definition \ref{Reproduction Number}, policymakers should actually increase public awareness about the epidemic, and more specifically, raising the lower bound $z_{\min}$ of $z$. In practice, this means implementing government prevention policies and promoting protective awareness through media and public health institutions, with the impact reflected in key epidemic parameters like the infection rate $\beta$ and the recovery rate $\delta$.

On the other hand, Proposition \ref{healthy-unstable} points out that a society that has not been affected by the epidemic or has already eradicated it (i.e., reached the equilibrium $s^* = (\mathbf{0}_n, \mathbf{0}_n)$) remains vulnerable when faced with a new severe epidemic with reproduction number $R^z_{\max} > 1$. This explains the phenomenon observed during the COVID-19 pandemic, where multiple waves of outbreaks occurred with the emergence of more infectious viral variants \cite{shrestha2022evolution}. This indicates that after an epidemic ends, policymakers must continuously implement effective public health measures and strengthen public health education, thereby maintaining a high level of public health awareness \cite{talic2021effectiveness}. This implies a higher $z_{\min}$ and a lower reproduction number $R^z_{\max}$ when facing the same epidemic.

\subsection{Existence and Stability Analysis of Endemic Equilibria}
The previous section thoroughly demonstrated the behavior of our epidemic-opinion coupled dynamics model as it converged to the healthy equilibrium. In this section, we aim to analyze the dynamic behavior and equilibria of more severe epidemics. Theorem \ref{healthy-consensus equilibrium} and Proposition \ref{healthy-unstable} have already shown that for a more severe epidemic with $R^z_{\max}>1$, the healthy equilibrium cannot be globally stable, meaning that the epidemic cannot always disappear spontaneously. Therefore, the existence and stability of endemic equilibria become the central focus of our analysis in this section.

First, we present the following proposition regarding the existence of an endemic equilibrium. Let us define $\Xi = [0,1]^{2 n} \backslash\left\{\left(\mathbf{0}_n, z\right) \mid z \in[0,1]^n\right\}$.
\begin{proposition} \label{endemic-existence}
	If $R^z_{\min}>1$, the coupled system (\ref{coupled}) has at least one endemic equilibrium $(x^*, z^*) \in \Xi$.
\end{proposition}
\begin{proof}
	From \cite[Proposition 1]{liu2019analysis}, $\rho\left(I_n-\Delta+B_{\min }\right)>1$ if and only if $s(-\Delta+B_{\min})>0$, where $s(M)$ denotes the largest real part among the eigenvalues of a real square matrix $M$. Let $\lambda \triangleq s\left(-\Delta+B_{\min }\right)$. Since $\Delta$ is diagonal and $B_{\min}$ is nonnegative and irreducible, $-\Delta+B_{\min}$ has an associated right eigenvector $\mu>0$ satisfying 
	\begin{equation}
		(-\Delta+B_{\min }) \mu = \lambda \mu \label{PF}
	\end{equation} 
	from the Perron--Frobenius theorem for irreducible Metzler matrices \cite[Lemma 2.3]{varga2009matrix}.
	
	Consider the convex and compact subset of $\Xi$ given by
	$$
	\Xi_{\epsilon} = \left\{\left(x, z\right) \mid x_i \in [\epsilon\mu_i, 1], z_i \in[0,1]^n, \forall i \in [n]\right\},
	$$
	where $\epsilon \in (0, 1)$. In what follows, we show that with some $\epsilon$, $\Xi_{\epsilon}$ is a positive invariant set for the system (\ref{coupled}). To this end, take $(x(k), z(k)) \in \Xi$. There must exist a sufficiently small $\epsilon_1$ such that $x_i(k) = \epsilon_1 \mu_i$ for some $i$, and $x_j(k) \in [\epsilon_1\mu_j, 1], \forall j \neq i$. Thus, $\left(x(k), z(k)\right) \in \Xi_{\epsilon_1}$. From (\ref{epidemic}) it follows that 
	\begin{align}
		x_i(k&+1) \nonumber\\
		= \ &\epsilon_1 \mu_i-\Big(\delta_{\min }+\left(\delta_i-\delta_{\min }\right) z_i(k)\Big) \epsilon_1 \mu_i \nonumber \\
		& +\Big(1-\epsilon_1 \mu_i\Big) \sum_{j \in \mathcal{N}_i}\Big(\beta_{i j}-\left(\beta_{i j}-\beta_{\min }\right) z_i(k)\Big) x_j(k) \nonumber\\
		\geq \ & \epsilon_1 \mu_i-\Big(\delta_{\min }+\left(\delta_i-\delta_{\min }\right) z_i(k)\Big) \epsilon_1 \mu_i \nonumber \\
		& +\Big(1-\epsilon_1 \mu_i\Big) \sum_{j \in \mathcal{N}_i}\Big(\beta_{i j}-\left(\beta_{i j}-\beta_{\min }\right) z_i(k)\Big) \epsilon_1\mu_j \nonumber\\
		= \ & \epsilon_1 \mu_i -\epsilon_1 \Big(\delta_{\min }+\left(\delta_i-\delta_{\min }\right) z_i(k)\Big)  \mu_i  \nonumber\\
		&+ \epsilon_1 \sum_{j \in \mathcal{N}_i}\Big(\beta_{i j}-\left(\beta_{i j}-\beta_{\min }\right) z_i(k)\Big)  \mu_j \nonumber \\
		&- \epsilon_1^2 \mu_i \sum_{j \in \mathcal{N}_i}\Big(\beta_{i j}-\left(\beta_{i j}-\beta_{\min }\right) z_i(k)\Big) \mu_j \nonumber\\
		\geq \ & \epsilon_1 \mu_i -\epsilon_1 \delta_i \mu_i  + \epsilon_1 \sum_{j \in \mathcal{N}_i}\beta_{\min}  \mu_j \nonumber\\
		&- \epsilon_1^2 \mu_i \sum_{j \in \mathcal{N}_i}\Big(\beta_{i j}-\left(\beta_{i j}-\beta_{\min }\right) z_i(k)\Big) \mu_j. \label{xik+1}
	\end{align}
	From (\ref{PF}), it holds that
	$\lambda \mu_i = -\delta_i \mu_i + \sum_{j \in \mathcal{N}_i} \beta_{\min} \mu_j.
	$
	Hence, (\ref{xik+1}) can be further written as
		\begin{align}
			x_i(k+1) &\geq \epsilon_1 \mu_i + \epsilon_1 \lambda \mu_i \nonumber\\
			&\quad - \epsilon_1^2 \mu_i \!\sum_{j \in \mathcal{N}_i}\!\Big(\beta_{i j}-\left(\beta_{i j}-\beta_{\min }\right) z_i(k)\Big) \mu_j .
		\end{align}
	Note that there must exist a sufficiently small $\epsilon_2$ such that
	$$
	\epsilon_2 \lambda \mu_i - \epsilon_2^2 \mu_i \sum_{j \in \mathcal{N}_i}\Big(\beta_{i j}-\left(\beta_{i j}-\beta_{\min }\right) z_i(k)\Big) \mu_j \geq 0,
	$$
	and then
	$
	x_i(k+1) \geq \epsilon_2 \mu_i.
	$
	Obviously, we obtain that $\forall \bar{\epsilon} \in (0, \min\{\epsilon_1, \epsilon_2\})$, $\Xi_{\bar{\epsilon}}$ is a positive invariant set for system (\ref{coupled}). Thus, by Brouwer's fixed-point theorem \cite{khamsi2011introduction}, system (\ref{coupled}) has at least one equilibrium in $\Xi_{\bar{\epsilon}}$, which completes the proof.
\end{proof}

Then, regarding the global asymptotic stability of the endemic equilibrium, we have the following result.
\begin{theorem} \label{endemic-stability}
	Suppose that $R^z_{\max}>1$ and the system (\ref{coupled}) has an endemic equilibrium $(x^*, z^*)$. Then $x^*$ is asymptotically stable for all disease-nonzero initial conditions, i.e., $x(0) \neq \mathbf{0}_n$, if
	\begin{equation} 
		\sum_{j \in \mathcal{N}_i} \beta_{i j} x_j^* \leq x_i^*, \ \forall i \in [n]. \label{beta x*} 
	\end{equation}
	Moreover, $z^*$ is globally asymptotically stable if 
	\begin{equation} 
		-2 w_{i i}-\frac{\theta_i}{1-\theta_i}<\left[\bar{L} z^*\right]_i<\frac{\theta_i}{1-\theta_i}, \forall i \in [n]. \label{Lz*} 
	\end{equation}
\end{theorem}
\begin{proof}
	According to infection process (\ref{opinion}), one has
	\begin{align}
		&\frac{\left(\delta_{\min}+\left(\delta_i-\delta_{\min}\right) z_i(k)\right) x_i^*}{1-x_i^*} \nonumber\\
		&\quad =\sum_{j \in \mathcal{N}_i}\left(\beta_{i j}-\left(\beta_{i j}-\beta_{\min }\right) z_i(k)\right) x_j^*. \label{xstar}
	\end{align}
	Let $y_i(k)=x_i(k)-x_i^*, \forall i \in[n]$. By (\ref{opinion}) and (\ref{xstar}), one has 
	\begin{equation}
		\begin{aligned}
			y_i(k+&1)\\
			= \ & y_i(k)-\left(\delta_{\min }+\left(\delta_i-\delta_{\min }\right) z_i(k)\right) y_i(k) \\
			& -y_i(k) \sum_{j \in \mathcal{N}_i}\left(\beta_{i j}-\left(\beta_{i j}-\beta_{\min }\right) z_i(k)\right) x_j^* \\
			& +\left(1-x_i(k)\right) \sum_{j \in \mathcal{N}_i}\left(\beta_{i j}-\left(\beta_{i j}-\beta_{\min }\right) z_i(k)\right) y_j(k) \\
			= & \bigg(1-\left(\delta_{\min }+\left(\delta_i-\delta_{\min }\right) z_i(k)\right) \\
			& -\sum_{j \in \mathcal{N}_i}\left(\beta_{i j}-\left(\beta_{i j}-\beta_{\min }\right) z_i(k)\right) x_j^*\bigg) y_i(k) \\
			& +\left(1-x_i(k)\right) \sum_{j \in \mathcal{N}_i}\left(\beta_{i j}-\left(\beta_{i j}-\beta_{\min }\right) z_i(k)\right) y_j(k) \\
			= & \bigg(1-\frac{1}{x_i^*} \sum_{j \in \mathcal{N}_i}\left(\beta_{i j}-\left(\beta_{i j}-\beta_{\min }\right) z_i(k)\right) x_j^*\bigg) y_i(k) \\
			& +\left(1-x_i(k)\right) \sum_{j \in \mathcal{N}_i}\left(\beta_{i j}-\left(\beta_{i j}-\beta_{\min }\right) z_i(k)\right) y_j(k).
		\end{aligned}
	\end{equation}
	A compact form of the above equation is given by
	\begin{equation}
		y(k+1)=\Phi(k) y(k), \label{time-varying}
	\end{equation}                      
	where
	$$
	\Phi(k)=I_n-H(k)+\operatorname{diag}(1-x(k))\left(B-Z(k)\left(B-B_{\min }\right)\right),
	$$
	in which
	$$
	\begin{aligned}
		H(k) & =\operatorname{diag}\bigg\{\bigg(\sum_{j \in \mathcal{N}_1}\left(\beta_{1 j}-\left(\beta_{1 j}-\beta_{\min }\right) z_1(k)\right) \frac{x_j^*}{x_1^*}, \\
		&  \qquad \ldots, \sum_{j \in \mathcal{N}_n}\left(\beta_{n j}-\left(\beta_{n j}-\beta_{\min }\right) z_n(k)\right) \frac{x_j^*}{x_n^*}\bigg)\bigg\}.
	\end{aligned}
	$$
	
	We now analyze the stability of the time-varying system (\ref{time-varying}). Let
	$
	F(k)=I-H(k)+\left(B-Z(k)\left(B-B_{\min }\right)\right).
	$
	As $B-Z(k)\left(B-B_{\min }\right)$ is an irreducible nonnegative matrix, it follows that $\Phi(k) \leq$ $F(k)$, where $\Phi(k)=F(k)$ holds if and only if $x(k)=0$.  Moreover, note that $I_n-H(k)$ is nonnegative by (\ref{beta x*}); hence, $\Phi(k)$ and $F(k)$ are irreducible nonnegative matrices. Now, denote
	$$
	\mu=\left[\begin{array}{llll}
		1 & \frac{x_2^*}{x_1^*} & \ldots & \frac{x_n^*}{x_1^*}
	\end{array}\right]^{\top}. 
	$$
	It is not difficult to verify from (\ref{xstar}) that $F(k) \mu=\mu-H \mu+\left(B-Z(k)\left(B-B_{\min }\right)\right) \mu=\mu$. Since $\mu > \mathbf{0}_n$, one obtains that $\mu$ is the Frobenius eigenvector of $F(k)$, corresponding to $\rho(F(k))=1$, by the Perron--Frobenius theorem for irreducible nonnegative matrices \cite[Corollary 8.1.30]{horn2012matrix}. Similarly, there exists a positive left eigenvector $v^{\top}$ such that $v^{\top} F(k)=v^{\top}$.
	
	To proceed, consider the following Lyapunov candidate:
	\begin{equation} \label{Lyapunov}
		V(k)=v^{\top}|y(k)|.
	\end{equation}
	We obtain
	\begin{align} \label{Vk}
		V&(k+1)-V(k) \leq v^{\top}(\Phi(k)-I_n)|y(k)| \nonumber\\
		& =v^{\top}(\Phi(k)-F(k))|y(k)| \nonumber\\
		& =-v^{\top} \operatorname{diag}(x(k))\left(B-Z(k)\left(B-B_{\min }\right)\right)|y(k)| \nonumber\\
		& \leq 0.
	\end{align}
	Note that $v > \mathbf{0}_n$, $B-Z(k)\left(B-B_{\min }\right)$ is irreducible nonnegative, and $x(k) > \mathbf{0}_n$ for $k > 0$ by Lemma \ref{well-posedness}. Hence the equality in (\ref{Vk}) holds if and only if $|y(k)|=0$. Therefore, the system in (\ref{time-varying}) is asymptotically stable for all disease-nonzero initial conditions. This holds regardless of the behavior of $z(k)$.
	
	On the other hand, let us consider the dynamics of $z(k)$. Denote $e_i(k)=z_i(k)-z^*$. By (\ref{opinion}), we have 
	$$
	\begin{aligned}
		e_i(&k+1)= \ z_i(k+1)-z_i^* \\
		= & \left(1-\theta_i\right)\Big(z_i(k)-z_i^* -\left(1-z_i^*\right) \sum_{j \in \mathcal{N}_i} w_{i j}\left(z_j^*-z_i^*\right) \\
		& +\left(1-z_i(k)\right) \sum_{j \in \mathcal{N}_i} w_{i j}\left(z_j(k)-z_i(k)\right)\Big) \\
		= & \left(1-\theta_i\right)\Big(e_i(k) -\left(1-z_i(k)+e_i(k)\right) \sum_{j \in \mathcal{N}_i} w_{i j}\left(z_j^*-z_i^*\right) \\
		& +\left(1-z_i(k)\right) \sum_{j \in \mathcal{N}_i} w_{i j}\left(e_j(k)+z_j^*-e_i(k)-z_i^*\right)\Big) \\
		= & \left(1-\theta_i\right)\Big(e_i(k)+\left(1-z_i(k)\right) \sum_{j \in \mathcal{N}_i} w_{i j}\left(e_j(k)-e_i(k)\right) \\
		& -e_i(k) \sum_{j \in \mathcal{N}_i} w_{i j}\left(z_j^*-z_i^*\right)\Big).
	\end{aligned}
	$$
	This can be written in the following compact form:
	$$
	\begin{aligned}
		e&(k+1) \\
		& =\left(I_n-\Theta\right)\left(e(k)-\left(I_n-Z(k)\right) \bar{L}e(k)+\operatorname{diag}\left(\bar{L} z^*\right) e(k)\right) \\
		& =\left(I_n-\Theta\right)\left(I_n-\left(I_n-Z(k)\right) \bar{L}+\operatorname{diag}\left(\bar{L} z^*\right)\right) e(k) \\
		& =\left(I_n-\Theta\right)\left(W+Z(k) \bar{L}+\operatorname{diag}\left(\bar{L} z^*\right)\right) e(k).
	\end{aligned}
	$$
	Now, let $M(k)=\left(I_n-\Theta\right)\left(W+Z(k) \bar{L}+\operatorname{diag}\left(\bar{L} z^*\right)\right)$. It thus follows that
	$$
	\begin{aligned}
		&\big[\left(I_n-\Theta\right)^{-1} M(k)\big]_{i j} \\
		& \quad = \begin{cases}w_{i i}+z_i(k)\left(1-w_{i i}\right)+\left[\bar{L} z^*\right]_i & (i=j), \\ w_{i j}-z_i(k) w_{i j}=\left(1-z_i(k)\right) w_{i j} & (i \neq j).\end{cases}
	\end{aligned}
	$$ 
	When $\left[\bar{L} z^*\right]_i \geq-w_{i i}-z_i(k)\left(1-w_{i i}\right)$, since $W$ is row-stochastic, one has 
	$$
	\begin{aligned}
		\sum_j&\left|M_{i j}(k)\right|  =\left(1-\theta_i\right)\bigg(\sum_j w_{i j}+\sum_j z_i(k) \bar{l}_{i j}+\left[\bar{L} z^*\right]_i\bigg) \\
		& =\left(1-\theta_i\right)\left(1+\left[\bar{L} z^*\right]_i\right) <\left(1-\theta_i\right)\left(1+\frac{\theta_i}{1-\theta_i}\right)=1,
	\end{aligned}
	$$
	where the inequality holds by (\ref{Lz*}). When $\left[\bar{L} z^*\right]_i<-w_{i i}-z_i(k)\left(1-w_{i i}\right)$, we have
	$$
	\begin{aligned}
		&\sum_j\left|M_{i j}(k)\right| =\left(1-\theta_i\right)\bigg(-w_{i i}-z_i(k)\left(1-w_{i i}\right) \\
		& \qquad \qquad \qquad \ -\left[\bar{L} z^*\right]_i +\sum_{j \in \mathcal{N}_i}\left(1-z_i(k)\right) w_{i j}\bigg) \\
		& \quad =\left(1-\theta_i\right)\left(1-2 w_{i i}-2 z_i(k)\left(1-w_{i i}\right)-\left[\bar{L} z^*\right]_i\right) \\
		& \quad <\left(1-\theta_i\right)\left(1-2 z_i(k)\left(1-w_{i i}\right)+\frac{\theta_i}{1-\theta_i}\right) \\
		& \quad =1-2 z_i(k)\left(1-w_{i i}\right)\left(1-\theta_i\right) \leq 1.
	\end{aligned}
	$$
	One then concludes that $\|M(k)\|_{\infty}<1, \forall k \geq 0$. Therefore, $z^*$ is asymptotically stable for all disease-nonzero initial conditions, which completes the proof.
\end{proof}

This proof draws inspiration from \cite[Theorem 2]{lin2021discrete}, where the Lyapunov function of the form (\ref{Lyapunov}) is proposed. The strength of our theorem lies in its incorporation of more realistic and complex polar opinion dynamics, providing more precise conditions for global stability. While existing research has extensively analyzed the stability of the healthy state in epidemics \cite{she2022networked,pare2020analysis,liu2019analysis}, results regarding endemic equilibrium are relatively scarce. Such works demonstrate the existence and uniqueness of the healthy state, but the existence and specific values of the endemic equilibria remain unknown, making the analysis significantly more challenging. In Theorem \ref{endemic-stability}, we avoid explicitly solving for the endemic equilibrium. Instead, by performing a coordinate transformation on the system states and considering the deviation from an equilibrium as the new state, we simplify the problem to resemble the stability analysis of the healthy equilibrium. This approach leads us to sufficient conditions for global stability, which are relatively easy to meet, as we discussed in the next Remark.

\begin{remark} \label{Rem:1}
	Condition (\ref{beta x*}) implies that, at the equilibrium, the infection level in each community is not lower than the combined influence from neighboring communities. Condition (\ref{Lz*}) suggests that the diversity in opinions among communities should not be excessive. These conditions are typically easy to satisfy in practice. For condition (\ref{beta x*}), although Assumption \ref{ass:1} imposes the constraint $\sum_{j=1}^n \beta_{i j} \leq 1$, actual infection rates $\sum_{j=1}^n \beta_{i j}$ are often much lower. For instance, even for highly infectious diseases like COVID-19, the value is estimated to be $\frac{3}{14}$ in \cite{morris2021optimal}. Condition (\ref{Lz*}) further implies that when polar opinion dynamics in (\ref{transformed-opinion}) has a strong effect (i.e., small $\theta$), opinions across communities in $z^*$ should be more aligned, whereas a larger $\theta$ allows a broader range for $z^*$. In fact, polar opinion dynamics naturally drives the opinions of communities toward consensus \cite{amelkin2017polar}, making condition (\ref{Lz*}) generally satisfied. Simulation results in the next section further confirm these observations.
\end{remark}

Proposition \ref{endemic-existence} and Theorem \ref{endemic-stability} jointly indicate that for a very severe infectious disease with $R_{\min }^z>1$, at least one endemic equilibrium $x^*$ always exists. When $x^*$ satisfies condition (\ref{beta x*}), it is asymptotically stable for any non-zero initial state. In practice, policymakers can use this equilibrium to guide medical resource redistribution. Since the existence of the endemic equilibrium is guaranteed by Brouwer's fixed-point theorem, various algorithms can be used for numerical computation \cite{karamardian2014fixed,chen2005algorithms}.

For moderately infectious diseases, where $R^z_{\min} \leq 1$ but $R^z_{\max}>1$, Theorem \ref{endemic-stability} ensures global convergence, assuming that an endemic equilibrium exists.  While empirical evidence and simulations in the next section support its existence, proving it rigorously is challenging due to the potential for  $R^z(k)$ in (\ref{Rz}) to repeatedly cross the threshold of $1$ as opinions evolve. This prevents the identification of a convex compact set containing only non-zero states. Therefore, it is hard to prove the existence of the endemic equilibrium. Proving the existence of such an equilibrium, or exploring conditions under which the system exhibits periodic oscillations or limit cycles, remains a compelling direction for future research.

\subsection{Discussion on Practical Epidemic Responses}
The previous sections have demonstrated the results of the autonomous co-evolution of the epidemic and opinion dynamics without any additional control strategies. Here, we discuss how policymakers can mitigate the epidemic's impact on the society by exerting external strategies, leveraging the insights gained from the preceding sections.

Recalling our SIS-opinion reproduction number given in Definition \ref{Reproduction Number}, we can show the following interesting property:

\begin{proposition} \label{monotonicty}
	For any $k_1$, $k_2$ $\geq 0$, if $z(k_1) \leq z(k_2)$, then $R^z_{k_1} \geq R^z_{k_2}$.
\end{proposition}
\begin{proof}
	We have $\Delta(z(k))=\operatorname{diag}\left\{\delta_{\text {min }}+\left(\delta_i-\delta_{\text {min }}\right) z_i(k)\right\}$ and $B(z(k))=\left[\beta_{i j}-\left(\beta_{i j}-\beta_{\text {min }}\right) z_i(k)\right]$ from Definition \ref{Reproduction Number}. When $z(k_1) \leq z(k_2)$, it follows that
	$$
	I_n-\Delta(z(k_1))+B(z(k_1)) \geq I_n-\Delta(z(k_2))+B(z(k_2)).
	$$
	Based on Assumption \ref{ass:1}, $I_n-\Delta(z(k))+B(z(k))$ is nonnegative and irreducible. Hence, from \cite[Theorem 2.7]{varga2009matrix}, we have 
	$$
	\rho(I_n-\Delta(z(k_1))+B(z(k_1))) \geq \rho(I_n-\Delta(z(k_2))+B(z(k_2))),
	$$
	which completes the proof.
\end{proof}

Then the following corollary for moderately infectious diseases (i.e., $R^z_{\min} \leq 1$ and $R^z_{\max}>1$) immediately follows from Proposition \ref{monotonicty}.

\begin{corollary} \label{cor:1}
	If $R^z_{\min} \leq 1$ and $R^z_{\max}>1$, there always exists an opinion vector $\tilde{z}$ such that $R^{\tilde{z}}(k) = 1$, where $z(k) \equiv \bar{z}$.
\end{corollary}	

Further, by following a proof similar to that of Theorem \ref{healthy-consensus equilibrium}, we obtain the next corollary.

\begin{corollary} \label{cor:2}
	If $R^z_{\min} \leq 1$ and $R^z_{\max}>1$, the healthy state $x^*=\mathbf{0}_n$ is asymptotically stable for all initial conditions if there exists $\bar{k}$ such that $\forall k \geq \bar{k}$, $z(k) \in [\bar{z}, \mathbf{1}_n]$, where $\bar{z}$ is the opinion vector given in Corollary \ref{cor:1}.
\end{corollary}

Corollaries \ref{cor:1} and \ref{cor:2} reveal the possibility of suppressing moderately infectious epidemics through opinion intervention strategies. To exploit this, we consider introducing external control inputs into the opinion dynamics in (7) as follows:
\begin{align}
	\begin{bmatrix}
		x(k+1) \\
		z(k+1)
	\end{bmatrix}=& \ \begin{bmatrix}
		\bar{K}(k) & 0_{n \times n} \\
		\Theta & \left(I_n-\Theta\right)(W+Z(k) \bar{L})
	\end{bmatrix} \begin{bmatrix}
		x(k) \\
		z(k)
	\end{bmatrix} \nonumber\\
	&+ \begin{bmatrix}
		0 & C
	\end{bmatrix} \begin{bmatrix}
		0 \\
		u
	\end{bmatrix}, 
\end{align}
where $C \in \mathbb{R}_{+}^{n \times m}$ is a nonnegative matrix, and $u \in \mathbb{R}_{+}^m$ is a constant input vector. We assume that under this control strategy, as $k \rightarrow \infty$, the infimum of the public opinion $z(k)$ is $\underline{z}^u$. It is easy to observe that $\underline{z}^u > C u \geq 0$.

For such a control problem, our focus is whether policymakers, given a fixed budget denoted by $Q$, can choose an optimal opinion intervention strategy for ensuring that $\underline{z}^u$ satisfies the conditions in Corollary \ref{cor:2}, thereby controlling the epidemic. Consider the following optimization problem:
\begin{equation}
	\max _{\sum_i u_i \leq Q} (\underline{z}^u)^{\top} \underline{z}^u.
\end{equation}
The optimal solution is denoted by $u^*(Q)$, with the corresponding maximum infimum opinion vector being $\underline{z}^u(Q)$.

In opinion dynamics, finding the optimal intervention strategy $u^*(Q)$ is an important area of research. Depending on different practical situations, the budget $Q$ can be allocated with specific weights across all or some nodes in the social network to induce an overall opinion shift \cite{damonte2022targeting}, or it can be used to place a few absolutely stubborn nodes, whose opinions do not change, at specific locations in the network to influence the entire system's opinion \cite{hunter2022optimizing}. It is not difficult to see that if $\underline{z}^u(Q) \geq \bar{z}$, the epidemic can be suppressed through opinion interventions. Otherwise, administrative measures such as lockdowns and mask mandates will be necessary \cite{cowling2020public}.

Finally in Algorithm \ref{alg:1}, we summarize the epidemic countermeasures. Algorithm \ref{alg:1} provides three response measures for policymakers depending on different infection levels: high, moderate, and mild. In the case of highly infectious diseases, opinion intervention strategies are ineffective in suppressing the epidemic. Administrative measures will be necessary. However, administrators can predict the infection proportions in each community in advance, allocate medical resources, and optimize response systems to minimize the epidemic's impact. For moderately infectious diseases, suppressing the epidemic through opinion interventions is a promising approach. Compared to administrative measures like lockdowns, it has a much smaller impact on the normal functioning of the society, provided that administrators have sufficient control over the opinion dynamics in the social network. Lastly, for mildly infectious diseases, Theorem \ref{healthy-consensus equilibrium} guarantees that the epidemic will naturally die out, so no special countermeasures are required.

\begin{algorithm}[!h]
	\caption{The algorithm for generating epidemic responses}
	\begin{algorithmic}[l]
		\REQUIRE Graphs $\mathcal{G}_S$ and $\mathcal{G}_P$, matrices $W$, $B$, $\Delta$, $\Theta$, and $Q$
		\ENSURE Epidemic responses
		\STATE Compute $R^z_{\min }$ and $R^z_{\max }$ by solving (\ref{Rzmin}) and (\ref{Rzmax}).
		\IF {$R^z_{\min } > 1$}
		\STATE Compute the globally stable endemic equilibrium $x^*$ guaranteed by Proposition \ref{endemic-existence} and Theorem \ref{endemic-stability}, using an algorithm for Brouwer fixed points \cite{karamardian2014fixed,chen2005algorithms}.
		\STATE Generate a medical resource redistribution strategy $\Psi$ based on $x^*$.
		\STATE Generate administrative measures set $\Omega$ (including lockdowns, mask mandates, and emergency declarations).
		\RETURN $\Psi$, $\Omega$
		\ELSIF {$R^z_{\min } \leq 1$ and $R^z_{\max } > 1$}
		\STATE Determine the opinion intervention strategy $u^*(Q)$ and the corresponding $\underline{z}^u(Q)$ via researches such as \cite{damonte2022targeting,hunter2022optimizing}. Then, set
		\STATE $R^{\underline{z}^u(Q)}_{\max } = \rho\left(I_n-\Delta\left(\underline{z}^u(Q)\right)+B\left(\underline{z}^u(Q)\right)\right)$.
		\IF {$R^{\underline{z}^u(Q)}_{\max } \leq 1$}
		\RETURN $u^*(Q)$
		\ELSE 
		\STATE Generate administrative measures set $\Omega$.
		\RETURN $\Omega$
		\ENDIF
		\ELSE
		\RETURN \textit{null}
		\ENDIF
	\end{algorithmic} \label{alg:1}
\end{algorithm}

\section{Simulations} \label{Sec:Simulation}
In this section, we simulate the epidemic spread over a large-scale real-world network based on 46 prefectures in Japan (excluding Kumamoto due to missing statistical data) and illustrate our main results. 

We consider an epidemic process spreading across Japan ($n=46$). Epidemic spread and public opinion evolution propagate through physical and social networks with different graph structures. Fig.~\ref{fig2a} shows our physical network, where the edges between prefectures represent population movement. The weights are determined based on the national migration survey\cite{national2019}. To simplify the network, many low-weight edges (representing less popular routes, such as Ehime-Iwate) were removed, as they are several orders of magnitude less important than busy routes (such as Tokyo-Kanagawa). The processed base adjacency matrix $\bar{B}$ remains irreducible. The base recovery rate matrix $\bar{\Delta}$ is constructed based on \cite{ministry2022}, reflecting regional medical resource distribution. Both $\bar{B}$ and $\bar{\Delta}$ are normalized to satisfy Assumption \ref{ass:1}. Fig.~\ref{fig2b} shows our social network, where the edges between prefectures represent public opinion exchange. To simulate the diverse social networks of the Internet age, we use the Watts-Strogatz small-world model \cite{watts1998collective} to generate a graph structure that satisfies Assumption \ref{ass:1}.

\begin{figure}[t!] 
	\centering
	\subfloat[\label{fig2a}]{
		\includegraphics[scale=0.50]{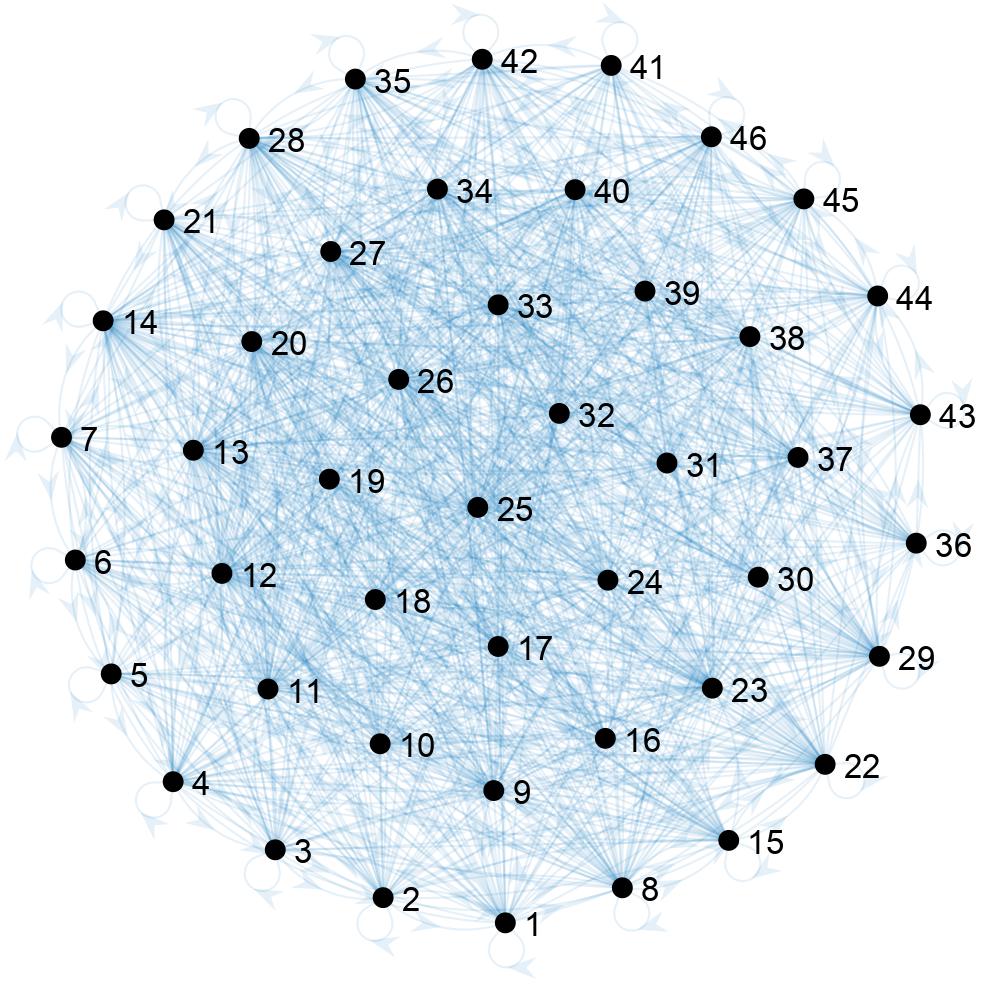}}
	\subfloat[\label{fig2b}]{
		\includegraphics[scale=0.58]{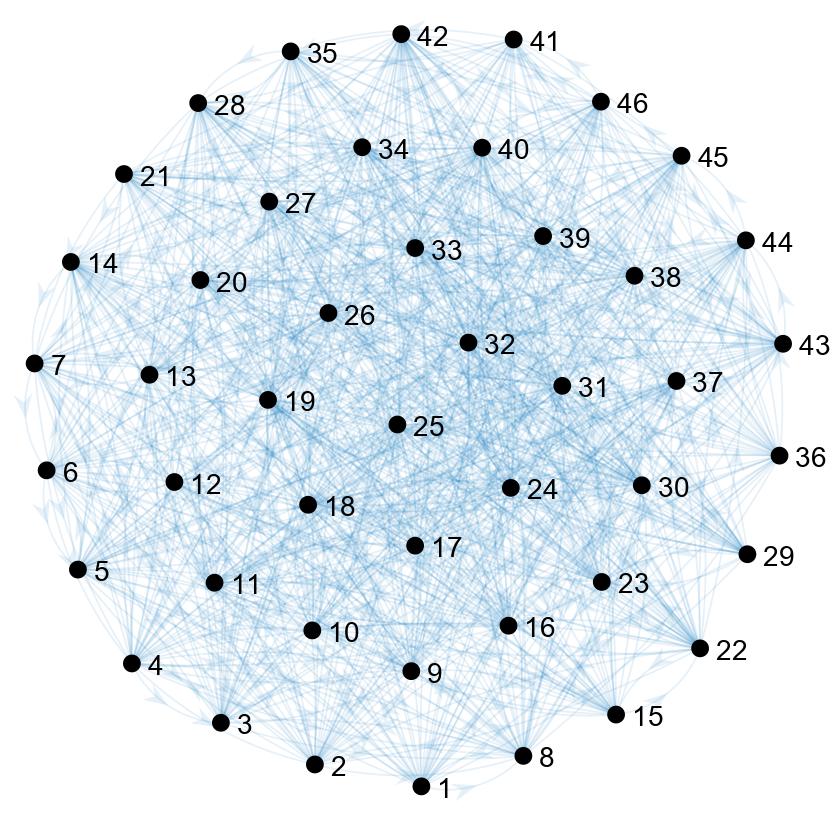} }
	\caption{Network structures. (a) Physical network for epidemic spreading. (b) Social network for opinion evolution.}\label{fig2}
\end{figure}

First, we consider a mild epidemic with $B=0.15 \bar{B}, B_{\min }=0.1 \bar{B}, \Delta=\bar{\Delta}$, and $\Delta_{\min }=$ $0.8 \bar{\Delta}$, resulting in $R_{\min }^z=0.9091$ and $R_{\max }^z=0.9927$. The development of the epidemic and public opinion is shown in Fig.~\ref{fig3}. Consistent with Theorem \ref{healthy-consensus equilibrium}, when $R_{\max }^z \leq 1$, i.e., in the case of a mild epidemic, the infection proportion in each prefecture converges to a disease-free healthy state, and the public quickly reaches a consensus, acknowledging that the epidemic is not severe. This type of infectious disease will spontaneously disappear, both physically and mentally, as described in Algorithm \ref{alg:1}, without the need for additional interventions.

\begin{figure} [t!]
	\centering
	\subfloat[\label{3a}]{
		\includegraphics[scale=0.25]{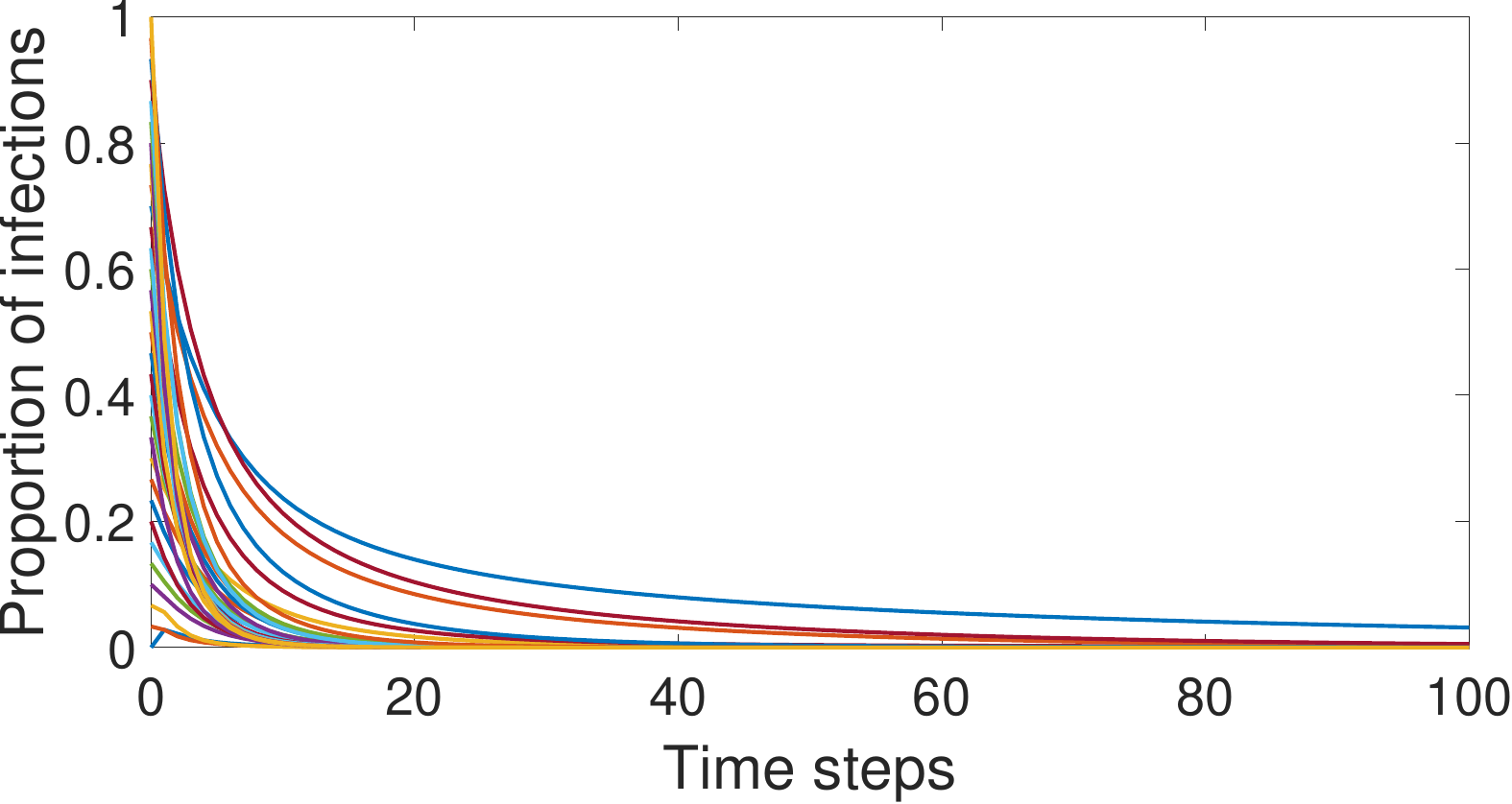}}
	\\
	\subfloat[\label{3b}]{
		\includegraphics[scale=0.25]{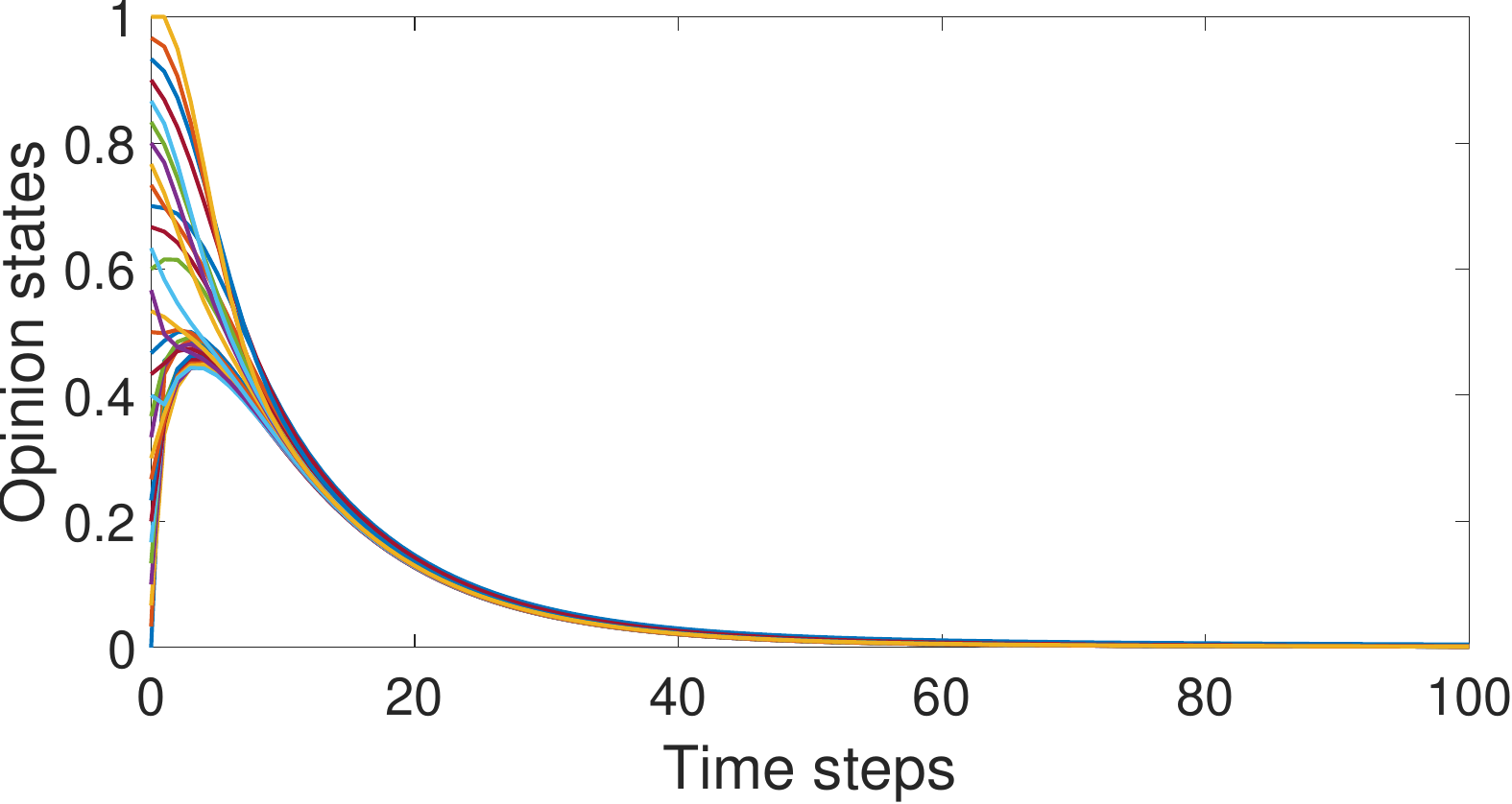} }
	\caption{For a mild epidemic with $R_{\min }^z=0.9091$ and $R_{\max }^z=0.9927$, evolution of the coupled system for the $46$-prefecture network in Fig.~\ref{fig2}: (a) epidemic states converge to a healthy equilibrium, and (b) opinion states reach consensus and converge to zero.}
	\label{fig3} 
\end{figure}

\begin{figure*} [t!]
	\centering
	\subfloat[\label{4a}]{
		\includegraphics[scale=0.25]{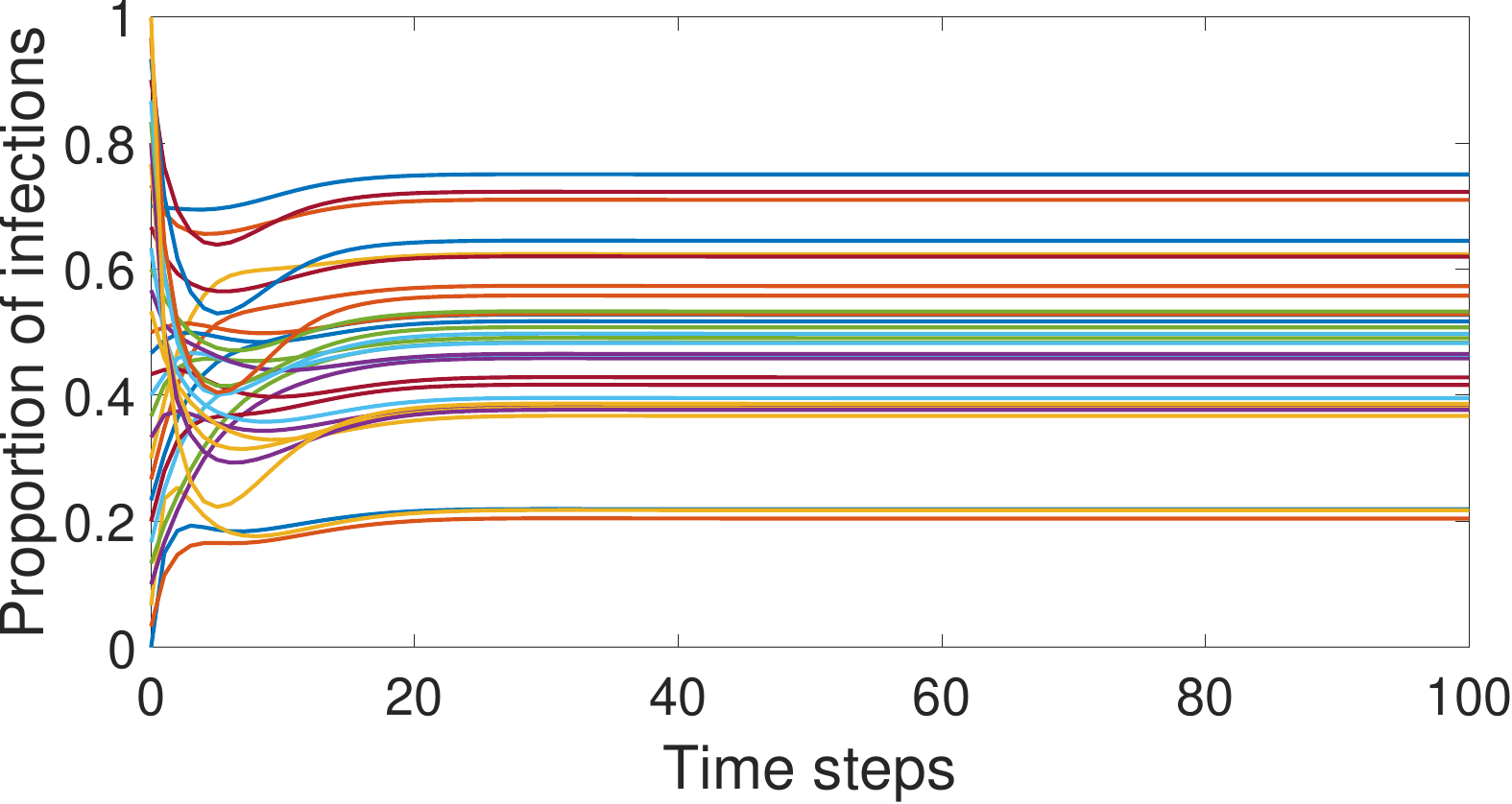}}
	\hspace{1cm}
	\subfloat[\label{4b}]{
		\includegraphics[scale=0.25]{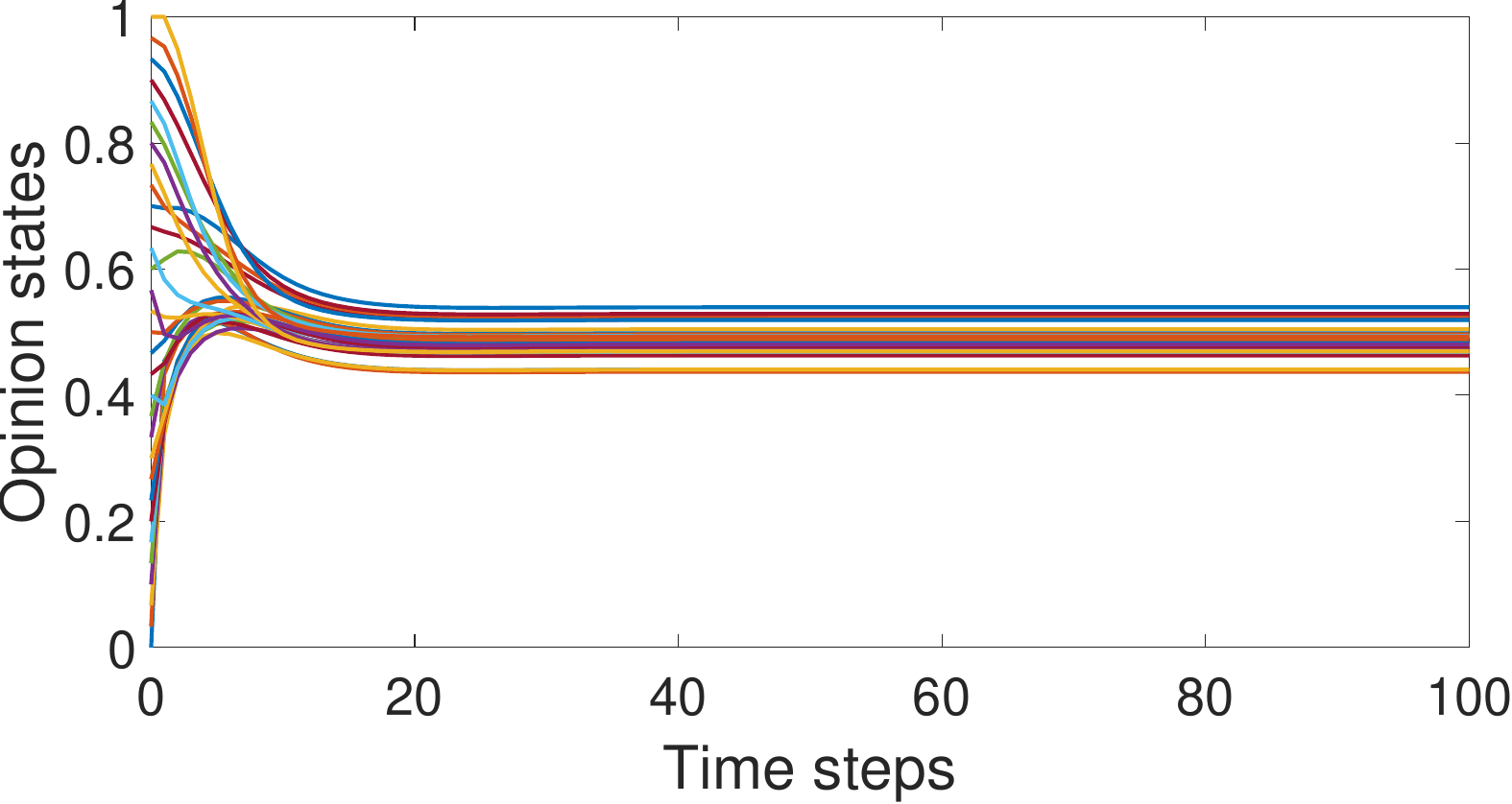} }
	\\
	\subfloat[\label{4c}]{
		\includegraphics[scale=0.25]{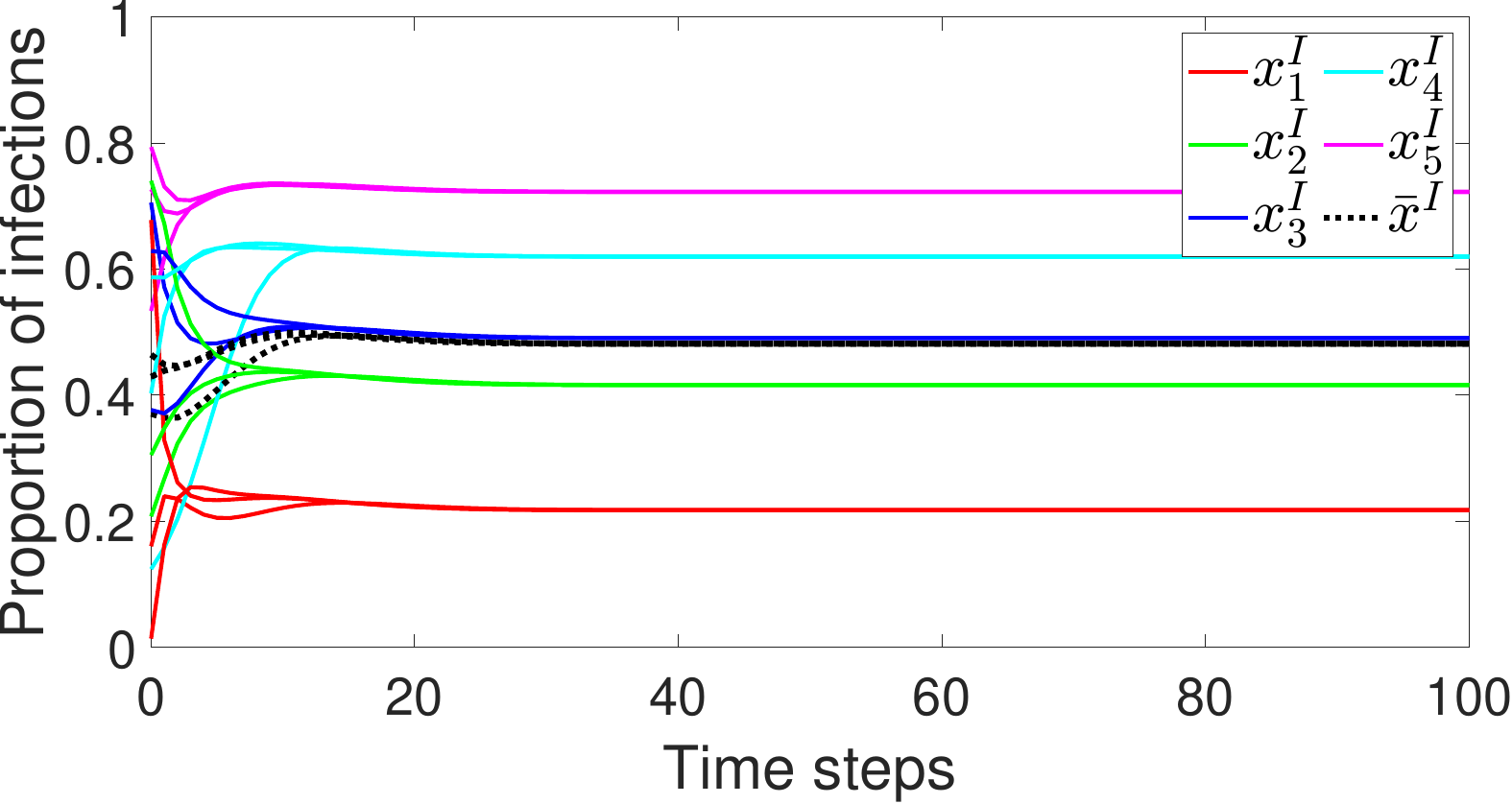}}
	\hspace{1cm}
	\subfloat[\label{4d}]{
		\includegraphics[scale=0.25]{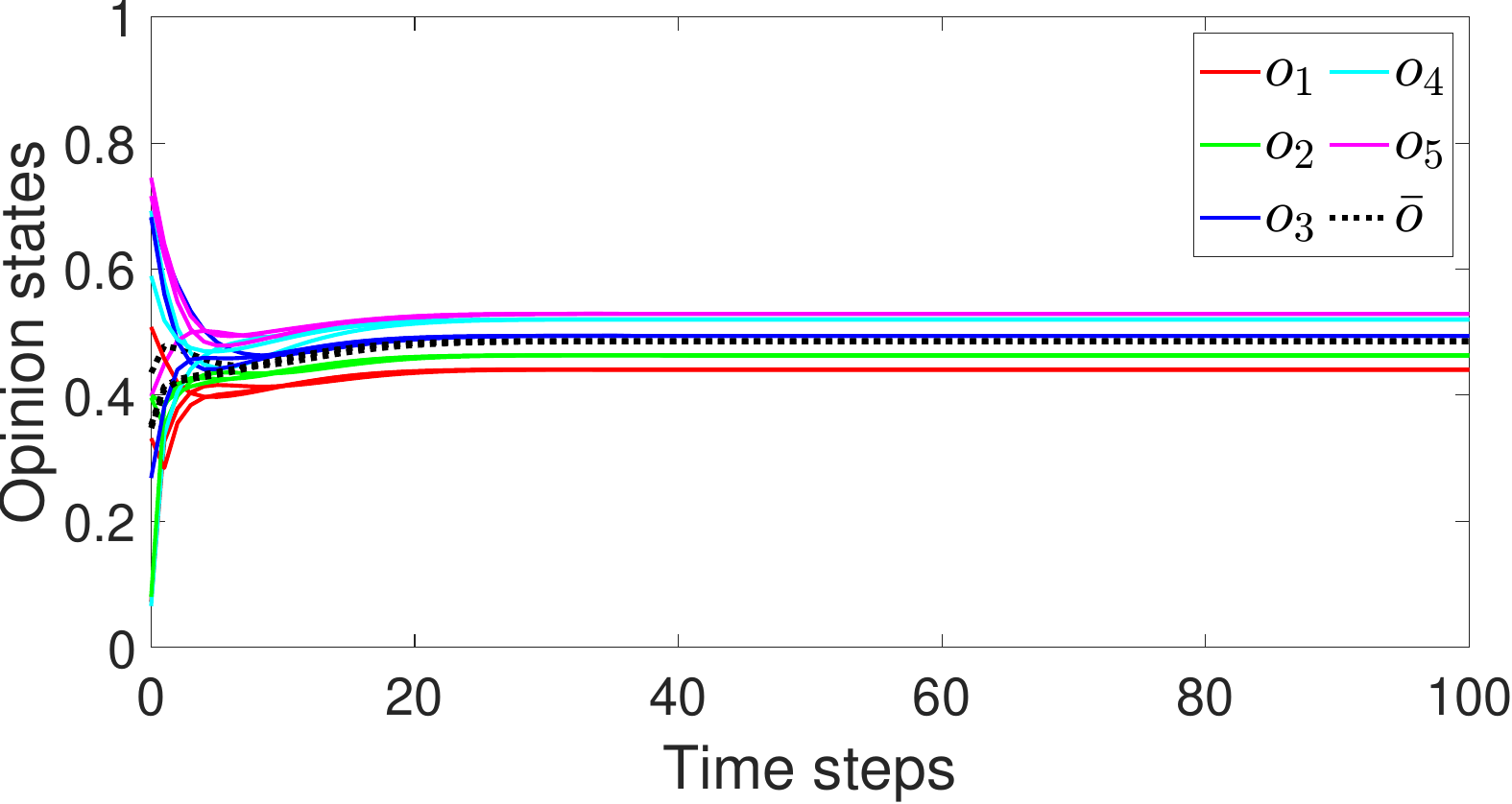} }
	\caption{For a severe epidemic with $R_{\min }^z=1.190$ and $R_{\max }^z=1.659$, evolution of the coupled system for the $46$-prefecture network in Fig.~\ref{fig2}: (a) epidemic states converge to an endemic equilibrium, (b) opinion states reach dissensus, (c) epidemic states converge to the same endemic equilibrium under 3 different initial conditions, and (d) opinion states reach the same dissensus under 3 different initial conditions. To maintain clarity, we only show $5$ randomly selected prefectures and the average values (thick black dotted lines) of all preferences in (c) and (d).}
	\label{fig4} 
\end{figure*}

Next, we consider a severe epidemic with $B=0.8 \bar{B}, B_{\min }=0.4 \bar{B}, \Delta=\bar{\Delta}$, and $\Delta_{\min }=0.5 \bar{\Delta}$, resulting in $R_{\min }^z=1.190$ and $R_{\max }^z=1.659$. As shown in Figs.~\ref{4a} and \ref{4b}, when $R_{\min }^z>1$, the coupled system (\ref{coupled}) has an endemic equilibrium $\left(x^*, z^*\right)$, confirming Proposition \ref{endemic-existence}. Moreover, we can verify that $x^*$ satisfies condition (\ref{beta x*}), and thus, Theorem~\ref{endemic-stability} indicates that $x^*$ should be globally asymptotically stable except at zero. Substituting $z^*$ into condition (\ref{Lz*}), we find that $z^*$ also meets the global asymptotic stability condition. The simulation in Figs.~\ref{4c} and \ref{4d} confirms these points: Regardless of the initial condition at the outbreak of the epidemic, the entire network ultimately converges to the same endemic equilibrium for this severe case. Therefore, as described in Algorithm \ref{alg:1}, solving (\ref{coupled}) in advance or obtaining $\left(x^*, z^*\right)$ through numerical methods allows for targeted allocation of medical resources and public health preparedness.

Finally, we consider a moderate epidemic with $B=0.2 \bar{B}, B_{\min }=0.1 \bar{B}, \Delta=\bar{\Delta}$, and $\Delta_{\min }=0.7 \bar{\Delta}$, resulting in $R_{\min }^z=0.9091$ and $R_{\max }^z=1.058$. Figs.~\ref{5a} and \ref{5b} show that compared to a severe epidemic, the system converges to an endemic equilibrium with a generally lower infection proportion. However, at the same time, the public's overall opinion toward the infectious disease becomes more optimistic, which is the key reason the epidemic cannot spontaneously disappear. At this point, we apply control strategies to the system, assuming that the government of each prefecture has sufficient administrative and promotional resources to raise the public opinion about the epidemic to $0.6$. In this case, $\underline{z}^u(Q)=0.6$ and the corresponding reproduction number $R_{\max }^{z^u(Q)}=0.9686<1$. With this strategy, Fig.~\ref{5c} shows that all prefectures successfully suppress the epidemic, achieving a healthy state. This demonstrates the results of Corollaries~\ref{cor:1} and \ref{cor:2} and Algorithm \ref{alg:1}, where government intervention in opinion dynamics can effectively control the epidemic with minimal impact on the socio-economic environment.

\begin{figure*} [t!]
	\centering
	\subfloat[\label{5a}]{
		\includegraphics[scale=0.25]{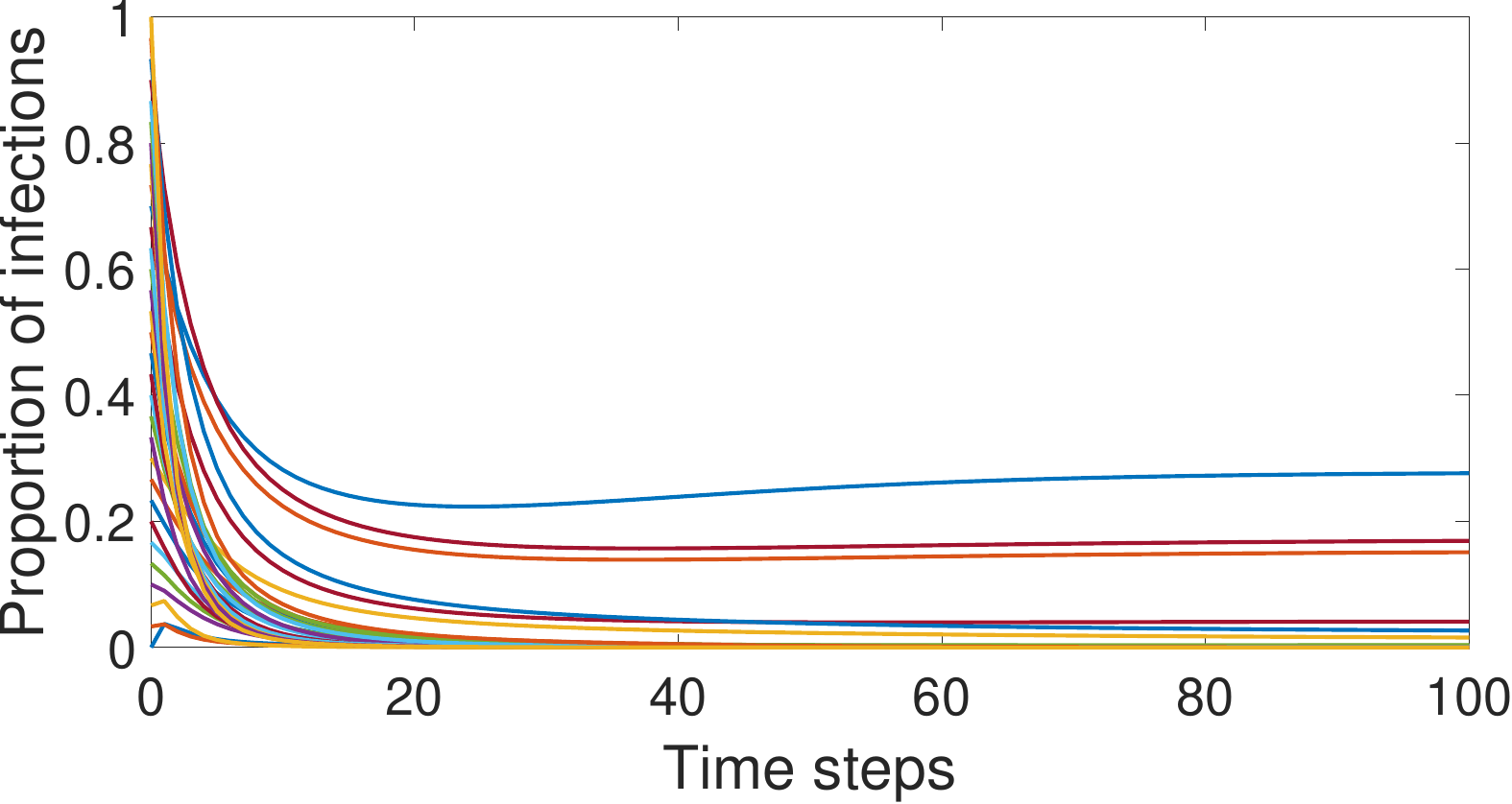}}
	\hspace{1cm}
	\subfloat[\label{5b}]{
		\includegraphics[scale=0.25]{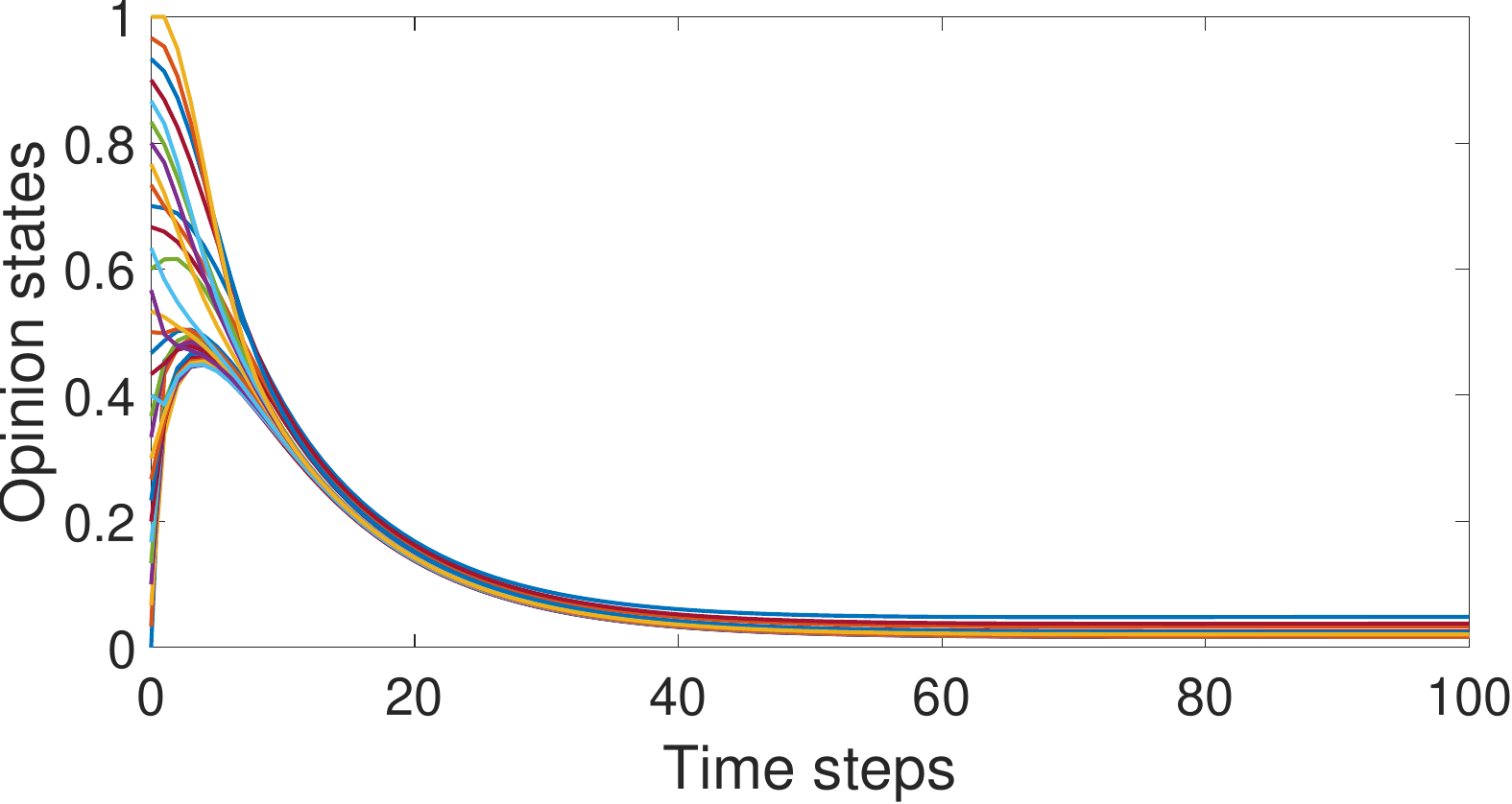} }
	\\
	\subfloat[\label{5c}]{
		\includegraphics[scale=0.25]{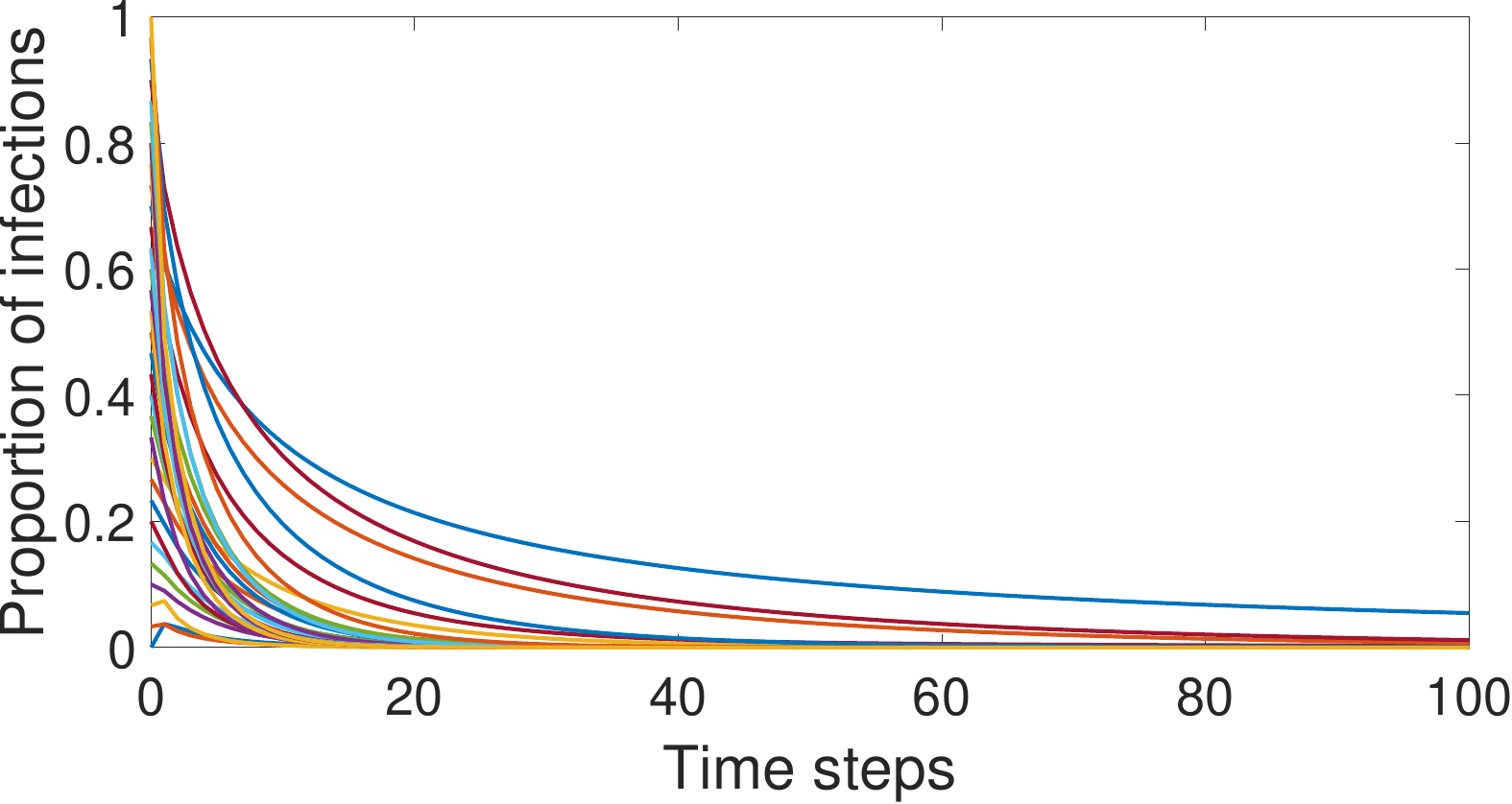}}
	\hspace{1cm}
	\subfloat[\label{5d}]{
		\includegraphics[scale=0.25]{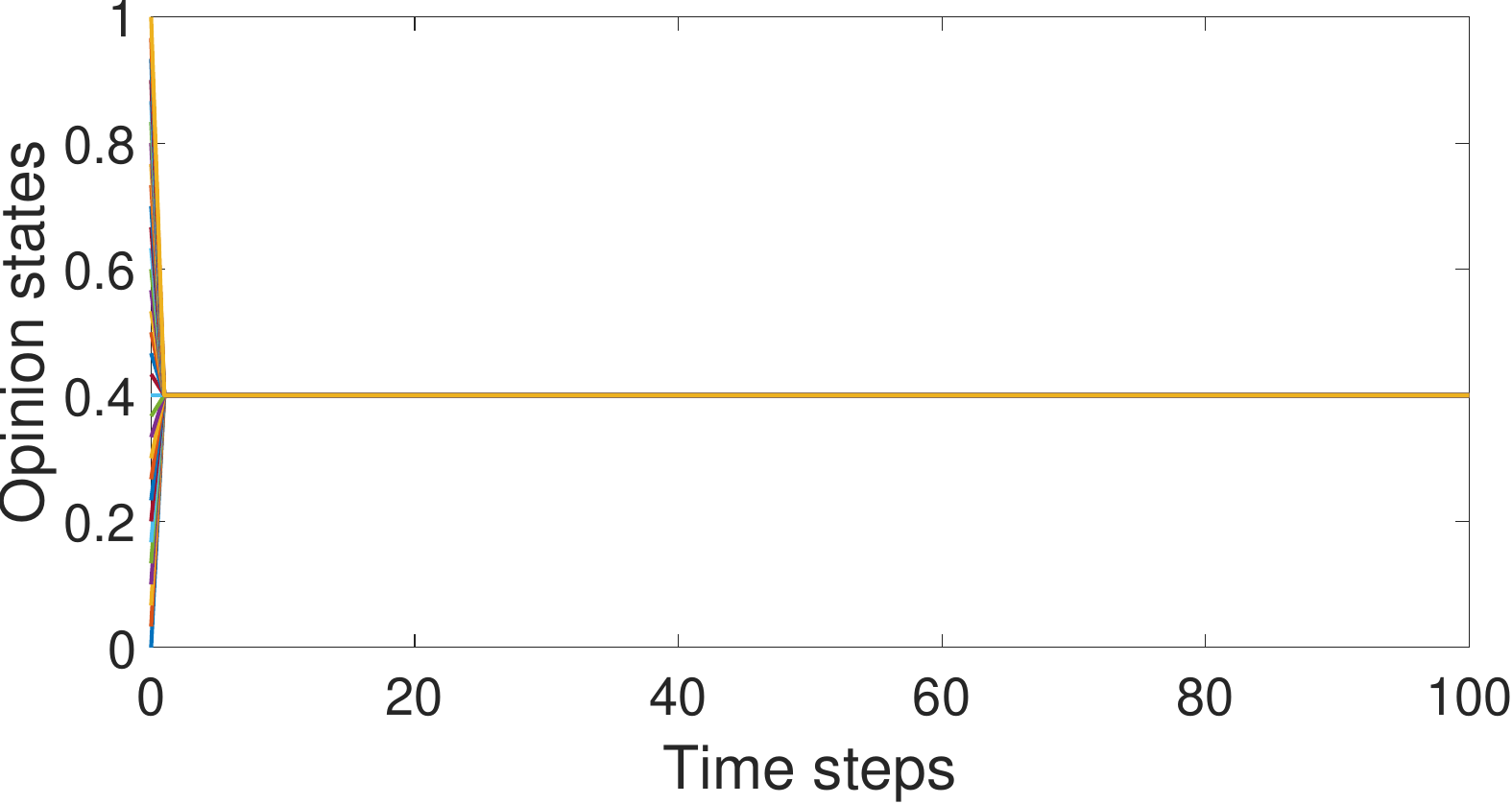} }
	\caption{For a moderate epidemic with $R_{\min }^z=0.9091$ and $R_{\max }^z=1.058$, evolution of the coupled system for the $46$-prefecture network in Fig.~\ref{fig2}: (a) epidemic states converge to an endemic equilibrium without intervention, (b) opinion states without intervention, (c) epidemic states converge to a healthy equilibrium with intervention, and (d) opinion states with intervention.}
	\label{fig5} 
\end{figure*}

\section{Conclusion} \label{Sec:Conclusion}
In this paper, we have considered the mutual influence between epidemic spread and public opinion evolution on a large-scale network. We have constructed a coupled epidemic-opinion model and, by defining an SIS-opinion reproduction number, obtained the necessary and sufficient conditions for the system to converge to a healthy equilibrium, as well as sufficient conditions for the existence and global stability of endemic equilibria. Our work has further discussed the role of public opinions in responding to large-scale epidemics and has proposed an algorithm to help policymakers devise strategies for dealing with epidemics of varying intensities and different available resources. A simulation example on a real-world network has supported our theoretical results and provided insights into epidemic control perspectives. 

In the future, we hope to extend our work to include more complex epidemic models, such as capturing the incubation period by an extra exposed state \cite{tomovski2021discrete} and considering the impact of non-biological environmental factors on epidemic spread through an infrastructure network layer \cite{pare2022multilayer}.

\bibliographystyle{IEEEtran}
\bibliography{root} 

\begin{thebibliography}{10}
\providecommand{\url}[1]{#1}
\csname url@samestyle\endcsname
\providecommand{\newblock}{\relax}
\providecommand{\bibinfo}[2]{#2}
\providecommand{\BIBentrySTDinterwordspacing}{\spaceskip=0pt\relax}
\providecommand{\BIBentryALTinterwordstretchfactor}{4}
\providecommand{\BIBentryALTinterwordspacing}{\spaceskip=\fontdimen2\font plus
\BIBentryALTinterwordstretchfactor\fontdimen3\font minus
  \fontdimen4\font\relax}
\providecommand{\BIBforeignlanguage}[2]{{%
\expandafter\ifx\csname l@#1\endcsname\relax
\typeout{** WARNING: IEEEtran.bst: No hyphenation pattern has been}%
\typeout{** loaded for the language `#1'. Using the pattern for}%
\typeout{** the default language instead.}%
\else
\language=\csname l@#1\endcsname
\fi
#2}}
\providecommand{\BIBdecl}{\relax}
\BIBdecl

\bibitem{pare2020analysis}
P.~E. Paré, J.~Liu, C.~L. Beck, B.~E. Kirwan, and T.~Basar, ``Analysis,
  estimation, and validation of discrete-time epidemic processes,'' \emph{IEEE
  Trans. on Control Syst. Technol.}, vol.~28, no.~1, pp. 79--93, 2020.

\bibitem{liu2019analysis}
J.~Liu, P.~E. Par{\'e}, A.~Nedi{\'c}, C.~Y. Tang, C.~L. Beck, and
  T.~Ba{\c{s}}ar, ``Analysis and control of a continuous-time bi-virus model,''
  \emph{IEEE Trans. Autom. Control}, vol.~64, no.~12, pp. 4891--4906, 2019.

\bibitem{she2022networked}
B.~She, J.~Liu, S.~Sundaram, and P.~E. Par{\'e}, ``On a networked {SIS}
  epidemic model with cooperative and antagonistic opinion dynamics,''
  \emph{IEEE Trans. Control Netw. Syst.}, vol.~9, no.~3, pp. 1154--1165, 2022.

\bibitem{wang2022resilient}
Y.~Wang, H.~Ishii, F.~Bonnet, and X.~D{\'e}fago, ``Resilient consensus for
  multi-agent systems under adversarial spreading processes,'' \emph{IEEE
  Trans. Netw. Sci. Eng.}, vol.~9, no.~5, pp. 3316--3331, 2022.

\bibitem{wu2023switching}
Z.~Wu, Y.~Wang, J.~Xiong, and M.~Xie, ``Switching periodic event-triggered
  {$H_{\infty}$} control for ncss and its application to uavs,'' \emph{IEEE
  Transactions on Vehicular Technology}, vol.~73, no.~4, pp. 5078--5088, 2024.

\bibitem{wu2019adaptive}
Z.~Wu, H.~Mo, J.~Xiong, and M.~Xie, ``Adaptive event-triggered observer-based
  output feedback {$L_{\infty}$} load frequency control for networked power
  systems,'' \emph{IEEE Transactions on Industrial Informatics}, vol.~16,
  no.~6, pp. 3952--3962, 2019.

\bibitem{wang2025TAC}
Y.~Wang, M.~Lv, Z.~Wu, Y.~Xu, G.~Nan, and R.~Su, ``Risk-constrained lqr design
  framework for non-gaussian interconnected systems defined over a digraph,''
  \emph{IEEE Transactions on Automatic Control}, vol.~70, no.~5, pp.
  3510--3517, 2025.

\bibitem{Xie2023auto}
K.~Xie, Z.~Wu, and J.~Xiong, ``Optimal periodic sensor scheduling for
  minimizing communication rate under lqg constraint,'' \emph{Automatica}, vol.
  156, p. 111191, 2023.

\bibitem{yu2025auto}
T.~Yu, J.~Song, Z.~Wu, and S.~He, ``Necessary and sufficient condition of
  distributed {$H_{\infty}$} filtering for interconnected large-scale systems:
  A novel space construction approach,'' \emph{Automatica}, vol. 171, p.
  111918, 2025.

\bibitem{zino2021analysis}
L.~Zino and M.~Cao, ``Analysis, prediction, and control of epidemics: A survey
  from scalar to dynamic network models,'' \emph{IEEE Circuits Syst. Mag.},
  vol.~21, no.~4, pp. 4--23, 2021.

\bibitem{randazzo2020sars}
W.~Randazzo, P.~Truchado, E.~Cuevas-Ferrando, P.~Sim{\'o}n, A.~Allende, and
  G.~S{\'a}nchez, ``{SARS-CoV-2} {RNA} in wastewater anticipated {COVID-19}
  occurrence in a low prevalence area,'' \emph{Water Res.}, vol. 181, p.
  115942, 2020.

\bibitem{karatas2022transportation}
M.~Karatas, L.~Eri{\c{s}}kin, and E.~Bozkaya, ``Transportation and location
  planning during epidemics/pandemics: Emerging problems and solution
  approaches,'' \emph{IEEE Trans. Intell. Transp. Syst.}, vol.~23, no.~12, pp.
  25\,139--25\,156, 2022.

\bibitem{teslya2022effect}
A.~Teslya, H.~Nunner, V.~Buskens, and M.~E. Kretzschmar, ``The effect of
  competition between health opinions on epidemic dynamics,'' \emph{PNAS
  Nexus}, vol.~1, no.~5, p. pgac260, 2022.

\bibitem{di2020covid}
F.~Di~Lauro, I.~Z. Kiss, D.~Rus, and C.~Della~Santina, ``{COVID-19} and
  flattening the curve: A feedback control perspective,'' \emph{IEEE Control
  Systems Letters}, vol.~5, no.~4, pp. 1435--1440, 2020.

\bibitem{paarporn2017networked}
K.~Paarporn, C.~Eksin, J.~S. Weitz, and J.~S. Shamma, ``Networked {SIS}
  epidemics with awareness,'' \emph{IEEE Trans. Comput. Soc. Syst.}, vol.~4,
  no.~3, pp. 93--103, 2017.

\bibitem{lin2021discrete}
Y.~Lin, W.~Xuan, R.~Ren, and J.~Liu, ``On a discrete-time network {SIS} model
  with opinion dynamics,'' in \emph{Proc. 60th IEEE Conference on Decision and
  Control}, 2021, pp. 2098--2103.

\bibitem{degroot1974reaching}
M.~H. DeGroot, ``Reaching a consensus,'' \emph{Journal of the American
  Statistical Association}, vol.~69, no. 345, pp. 118--121, 1974.

\bibitem{wang2021suppressing}
Y.~Wang, S.~Gracy, H.~Ishii, and K.~H. Johansson, ``Suppressing the endemic
  equilibrium in {SIS} epidemics: A state dependent approach,''
  \emph{IFAC-PapersOnLine}, vol.~54, no.~15, pp. 163--168, 2021.

\bibitem{amelkin2017polar}
V.~Amelkin, F.~Bullo, and A.~K. Singh, ``Polar opinion dynamics in social
  networks,'' \emph{IEEE Trans. on Autom. Control}, vol.~62, no.~11, pp.
  5650--5665, 2017.

\bibitem{organisation2022first}
{Organisation for Economic Co-operation and Development}, \emph{First Lessons
  from Government Evaluations of COVID-19 Responses: A Synthesis}.\hskip 1em
  plus 0.5em minus 0.4em\relax OECD Publishing, 2022.

\bibitem{yu2024individuals}
G.~Yu, M.~Garee, M.~Ventresca, and Y.~Yih, ``How individuals’ opinions
  influence society’s resistance to epidemics: {An} agent-based model
  approach,'' \emph{BMC Public Health}, vol.~24, no.~1, p. 863, 2024.

\bibitem{mawson2017mass}
A.~R. Mawson, \emph{Mass Panic and Social Attachment: The Dynamics of Human
  Behavior}.\hskip 1em plus 0.5em minus 0.4em\relax Routledge, 2017.

\bibitem{xu2024analysis}
Q.~Xu, T.~Masada, and H.~Ishii, ``Analysis of stubborn opinions on networked
  sis epidemic dynamics,'' submitted for conference publication, 2024.

\bibitem{proskurnikov2017tutorial}
A.~V. Proskurnikov and R.~Tempo, ``A tutorial on modeling and analysis of
  dynamic social networks. {Part I},'' \emph{Annu. Rev. Control}, vol.~43, pp.
  65--79, 2017.

\bibitem{miller1993attitude}
A.~G. Miller, J.~W. McHoskey, C.~M. Bane, and T.~G. Dowd, ``The attitude
  polarization phenomenon: Role of response measure, attitude extremity, and
  behavioral consequences of reported attitude change.'' \emph{J. Personal.
  Soc. Psychol.}, vol.~64, no.~4, p. 561, 1993.

\bibitem{sjoberg2020explaining}
L.~Sj{\"o}berg, ``Explaining risk perception: An empirical evaluation of
  cultural theory,'' in \emph{Risk Management}, G.~Mars and D.~T. Weir,
  Eds.\hskip 1em plus 0.5em minus 0.4em\relax Routledge, 2020, pp. 127--144.

\bibitem{lu2023daily}
G.~Lu, Z.~Yang, W.~Qu, T.~Qian, Z.~Liu, W.~He, Z.~Lin, and C.~Hon, ``Daily
  fluctuations in {COVID-19} infection rates under {Tokyo}’s epidemic
  prevention measures---{New} evidence from adaptive {Fourier} decomposition,''
  \emph{Frontiers in Public Health}, vol.~11, p. 1245572, 2023.

\bibitem{green2020health}
E.~C. Green, E.~M. Murphy, and K.~Gryboski, ``The health belief model,'' in
  \emph{The Wiley Encyclopedia of Health Psychology}, R.~H. Paul, L.~E.
  Salminen, J.~Heaps, and L.~M. Cohen, Eds.\hskip 1em plus 0.5em minus
  0.4em\relax Wiley Online Library, 2020, pp. 211--214.

\bibitem{ajzen2020theory}
I.~Ajzen, ``The theory of planned behavior: Frequently asked questions,''
  \emph{Hum. Behav. Emerg. Technol.}, vol.~2, no.~4, pp. 314--324, 2020.

\bibitem{szczuka2021trajectory}
Z.~Szczuka, C.~Abraham, A.~Baban, S.~Brooks, S.~Cipolletta, E.~Danso, S.~U.
  Dombrowski, Y.~Gan, T.~Gaspar, M.~G. de~Matos \emph{et~al.}, ``The trajectory
  of {COVID-19} pandemic and handwashing adherence: {Findings} from 14
  countries,'' \emph{BMC Public Health}, vol.~21, pp. 1--13, 2021.

\bibitem{shrestha2022evolution}
L.~B. Shrestha, C.~Foster, W.~Rawlinson, N.~Tedla, and R.~A. Bull, ``Evolution
  of the {SARS-CoV-2} omicron variants {BA.1} to {BA.5}: Implications for
  immune escape and transmission,'' \emph{Reviews in Medical Virology},
  vol.~32, no.~5, p. e2381, 2022.

\bibitem{talic2021effectiveness}
S.~Talic, S.~Shah, H.~Wild, D.~Gasevic, A.~Maharaj, Z.~Ademi, X.~Li, W.~Xu,
  I.~Mesa-Eguiagaray, J.~Rostron \emph{et~al.}, ``Effectiveness of public
  health measures in reducing the incidence of {COVID-19}, {SARS-CoV-2}
  transmission, and {COVID-19} mortality: {Systematic} review and
  meta-analysis,'' \emph{BMJ}, vol. 375, 2021.

\bibitem{varga2009matrix}
R.~S. Varga, \emph{Matrix Iterative Analysis}.\hskip 1em plus 0.5em minus
  0.4em\relax Springer, 2009.

\bibitem{khamsi2011introduction}
M.~A. Khamsi and W.~A. Kirk, \emph{An Introduction to Metric Spaces and Fixed
  Point Theory}.\hskip 1em plus 0.5em minus 0.4em\relax John Wiley \& Sons,
  2011.

\bibitem{horn2012matrix}
R.~A. Horn and C.~R. Johnson, \emph{Matrix Analysis}.\hskip 1em plus 0.5em
  minus 0.4em\relax Cambridge University Press, 2012.

\bibitem{morris2021optimal}
D.~H. Morris, F.~W. Rossine, J.~B. Plotkin, and S.~A. Levin, ``Optimal,
  near-optimal, and robust epidemic control,'' \emph{Communications Physics},
  vol.~4, no.~1, p.~78, 2021.

\bibitem{karamardian2014fixed}
S.~Karamardian, \emph{Fixed Points: Algorithms and Applications}.\hskip 1em
  plus 0.5em minus 0.4em\relax Academic Press, 2014.

\bibitem{chen2005algorithms}
X.~Chen and X.~Deng, ``On algorithms for discrete and approximate {Brouwer}
  fixed points,'' in \emph{Proc. 37th ACM Symposium on Theory of Computing},
  2005, pp. 323--330.

\bibitem{damonte2022targeting}
L.~Damonte, G.~Como, and F.~Fagnani, ``Targeting interventions for displacement
  minimization in opinion dynamics,'' in \emph{Proc. 50th IEEE Conf. on Deci.
  and Control}, 2022, pp. 7023--7028.

\bibitem{hunter2022optimizing}
D.~S. Hunter and T.~Zaman, ``Optimizing opinions with stubborn agents,''
  \emph{Operations Research}, vol.~70, no.~4, pp. 2119--2137, 2022.

\bibitem{cowling2020public}
B.~J. Cowling and A.~E. Aiello, ``Public health measures to slow community
  spread of coronavirus disease 2019,'' \emph{The Journal of Infectious
  Diseases}, vol. 221, no.~11, pp. 1749--1751, 2020.

\bibitem{national2019}
\BIBentryALTinterwordspacing
{National Institute of Population and Social Security Research of Japan}.
  (2019) {The 8th National Survey on Migration}. [Online]. Available:
  \url{https://www.ipss.go.jp/ps-idou/e/m08e/mig08e.asp}
\BIBentrySTDinterwordspacing

\bibitem{ministry2022}
\BIBentryALTinterwordspacing
{Ministry of Health, Labour and Welfare of Japan}. (2022) {The 4th Working
  Group on Regional Medical Vision and Physician Securing Plan}. [Online].
  Available: \url{https://www.mhlw.go.jp/stf/newpage\_25551.html}
\BIBentrySTDinterwordspacing

\bibitem{watts1998collective}
D.~J. Watts and S.~H. Strogatz, ``Collective dynamics of ‘small-world’
  networks,'' \emph{Nature}, vol. 393, no. 6684, pp. 440--442, 1998.

\bibitem{tomovski2021discrete}
I.~Tomovski, L.~Basnarkov, and A.~Abazi, ``Discrete-time non-{Markovian} {SEIS}
  model on complex networks,'' \emph{IEEE Trans. Netw. Sci. Eng.}, vol.~9,
  no.~2, pp. 552--563, 2021.

\bibitem{pare2022multilayer}
P.~E. Par{\'e}, A.~Janson, S.~Gracy, J.~Liu, H.~Sandberg, and K.~H. Johansson,
  ``Multilayer {SIS} model with an infrastructure network,'' \emph{IEEE Trans.
  Control Netw. Syst.}, vol.~10, no.~1, pp. 295--307, 2022.

\end{thebibliography}

\begin{IEEEbiography}[{\includegraphics[width=1in,height=1.25in,clip,keepaspectratio]{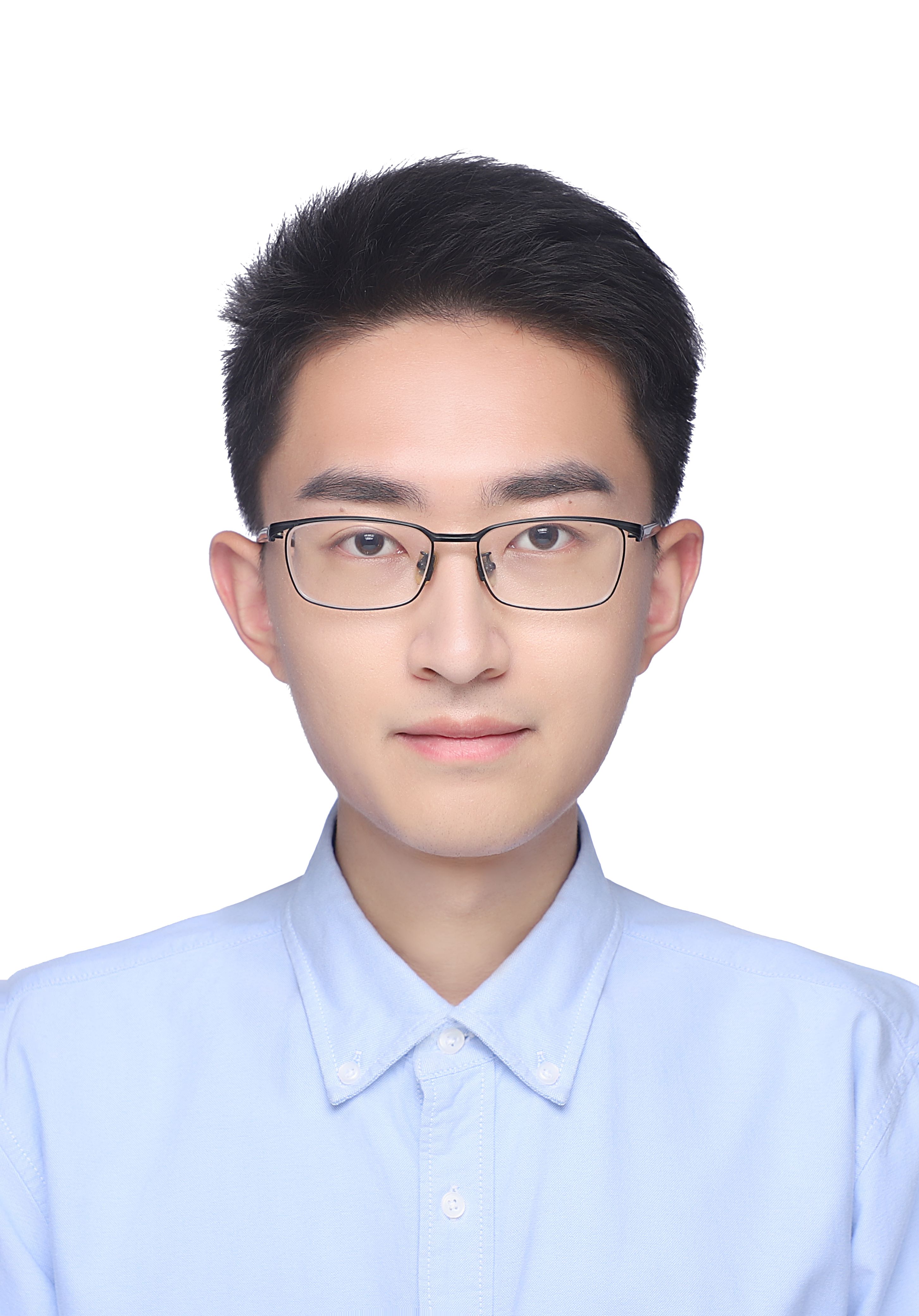}}]{Qiulin Xu}
	received the B.Eng. degree in automation from the School of Artificial Intelligence and Automation, Huazhong University of Science and Technology, Wuhan, China, in 2019, and the M.Eng. degree in control engineering from the Department of Automation, University of Science and Technology of China, Hefei, China, in 2022. He is currently working toward the Ph.D degree in artificial intelligence with the Institute of Science Tokyo (formerly, Tokyo
	Institute of Technology, until September 2024), Yokohama, Japan. His current research interests include cyber-physical systems security, event-based state estimation, epidemic modeling and control, and social networks.
\end{IEEEbiography}

\begin{IEEEbiography}[{\includegraphics[width=1in,height=1.25in,clip,keepaspectratio]{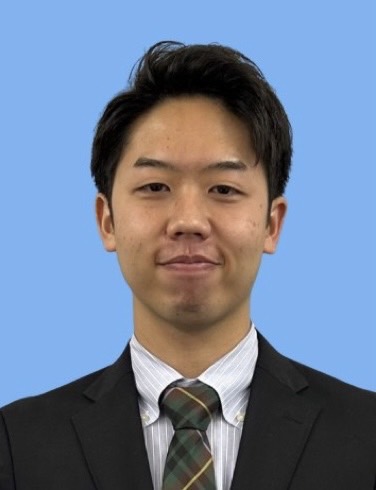}}]{Tatsuya Masada}
	received the B.Eng. degree in Computer Science from the Tokyo Institute of  Technology, Tokyo, Japan, in 2022, and the M.Eng. degree in Computer Science from the Tokyo Institute of  Technology, Tokyo, Japan, in 2024. He is currently working on information security at an IT company. His research interests include epidemic networks.
\end{IEEEbiography}

%\begin{IEEEbiography}[{\includegraphics[width=1in,height=1.25in,clip,keepaspectratio]{Biography/Xu.jpg}}]{Tatsuya Masada}
%	received 
%\end{IEEEbiography}

\begin{IEEEbiography}[{\includegraphics[width=1in,height=1.25in,clip,keepaspectratio]{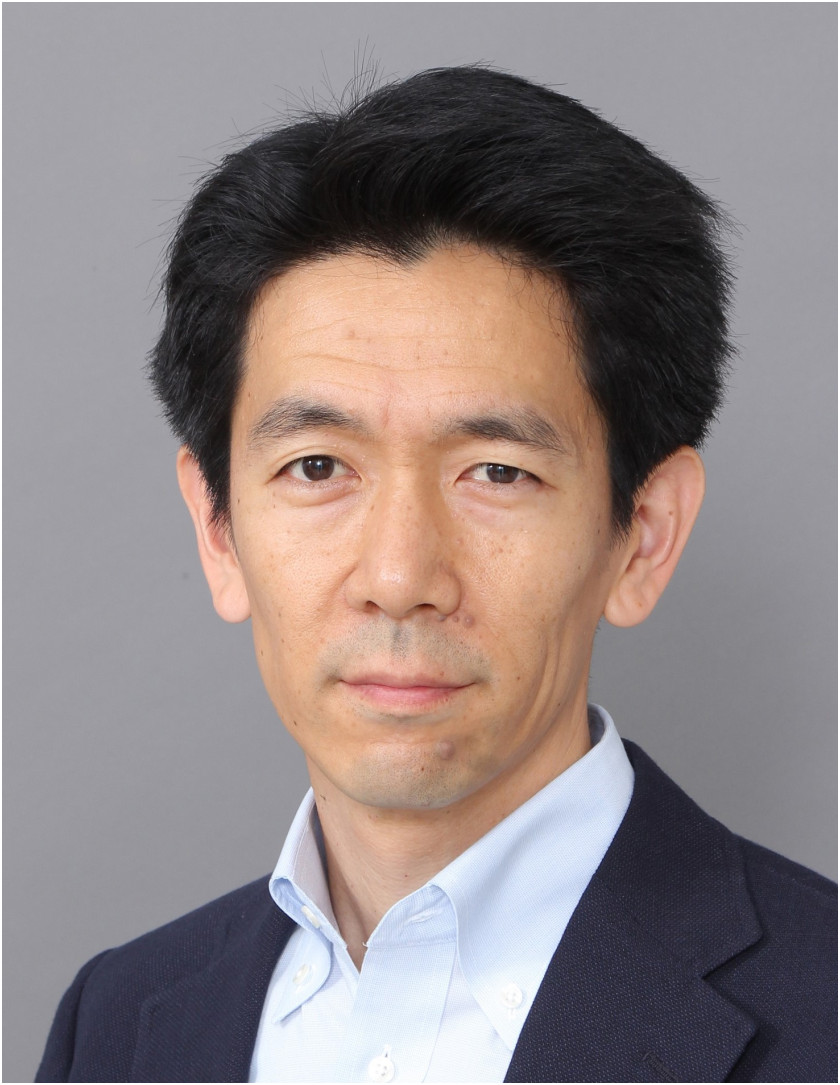}}]{Hideaki Ishii}
	(M'02-SM'12-F'21) received the
	M.Eng.\ degree from Kyoto University in 1998, 
	and the Ph.D.\ degree from the University of Toronto in 2002. He was a Postdoctoral Research
	Associate at the University of Illinois at Urbana-Champaign in 2001--2004, and a Research Associate at The University of Tokyo in 2004--2007.
	He was an Associate Professor and then a Professor 
	at the Department of Computer Science,
	Tokyo Institute of Technology in 2007--2024.
	Currently, he is a Professor at the Department of
	Information Physics and Computing at The University of Tokyo since 2024.
	He was a Humboldt Research Fellow at the University of Stuttgart
	in 2014--2015. He has also held visiting positions at CNR-IEIIT at
	the Politecnico di Torino, the Technical University of Berlin, and
	the City University of Hong Kong. His research interests
	include networked control systems, multiagent systems, distributed algorithms,
	and cyber-security of control systems.
	
	Dr.~Ishii has served as an Associate Editor for Automatica, 
	the IEEE Control Systems Letters, the IEEE Transactions on Automatic Control, 
	the IEEE Transactions on Control of Network Systems,
	and the Mathematics of Control, Signals, and Systems.
	He was a Vice President for the IEEE Control Systems Society (CSS) in 2022--2023 
	and an elected member of the IEEE CSS Board of Governors in 2014--2016. 
	He was the Chair of the IFAC Coordinating Committee on Systems and
	Signals in 2017--2023 and the Chair of the IFAC Technical Committee
	on Networked Systems for 2011--2017. 
	He served as the IPC Chair for the IFAC World Congress 2023 held in Yokohama, Japan.
	He received the IEEE Control Systems Magazine Outstanding Paper
	Award in 2015. Dr.~Ishii is an IEEE Fellow.
\end{IEEEbiography}

\end{document}